%% file: main.tex
\documentclass{article}

\usepackage{style}
\usepackage{fullpage}
\usepackage[symbol]{footmisc}
\usepackage{caption}
\usepackage{color}
\usepackage[mathcal]{euscript}
\usepackage{enumitem}
\usepackage{comment}
\usepackage{hyperref}
\usepackage{authblk}
\usepackage{times}

\usepackage{url}
\date{}

\usepackage{microtype}

\input{macro}

\newcommand{\BF}{\mathcal{BF}}
\newcommand{\NC}{\mathcal{NC}}
\newcommand{\AC}{\mathcal{AC}}
\newcommand{\LogSpace}{\mathcal{L}}
\newcommand{\PRAM}{\mathcal{PRAM}}

\def\fullversion{1}

\begin{document}
%\fancyhead{}

%\title[Optimal Parallel Algorithms in the Binary-Forking Model]{Optimal (Randomized) Parallel Algorithms \\in the Binary-Forking Model}
\title{Optimal (Randomized) Parallel Algorithms \\in the Binary-Forking Model}

\author{Guy E. Blelloch\\ guyb@cs.cmu.edu\\ Carnegie Mellon Universtiy \and Jeremy T. Fineman \\ jfineman@cs.georgetown.edu \\ Georgetown University \and Yan Gu \\ ygu@cs.ucr.edu \\ University of California, Riverside \and Yihan Sun \\ yihans@cs.ucr.edu \\ University of California, Riverside}

\maketitle

\input{abstract}

\input{intro}
\input{nmodels}

\input{listcontraction}

\input{sorting}
\input{setset/setset}

\input{rand}
\input{figure-tree}
\input{rmq}

\input{treecontraction}

\section*{Acknowledgement}

This work was supported in part by NSF grants CCF-1408940, CCF-1629444, CCF-1718700, CCF-1910030, CCF-1918989, and CCF-1919223.

\bibliographystyle{plain}
\bibliography{bibliography/strings,bibliography/main}

\ifx\fullversion\undefined
\else
\appendix
\input{setset/base}

\input{setset/proof}

\input{app-others}

\fi

\end{document}

%% file: macro.tex
\newcommand{\forkins}{\texttt{fork}}
\newcommand{\insend}{\texttt{end}}
\newcommand{\finish}{\texttt{endall}}
\newcommand{\allocateins}{\texttt{allocate}}
\newcommand{\freeins}{\texttt{free}}

\newcommand{\thread}{thread}
\newcommand{\bfmodel}{binary-forking model}
\newcommand{\tas}{\texttt{TS}}
\newcommand{\cas}{\texttt{CAS}}
\newcommand{\tszero}{\textbf{zero}}
\newcommand{\tsone}{\textbf{one}}
\newcommand{\fjmodel}{binary fork-join model}
\newcommand{\join}{join}
\newcommand{\joinins}{\texttt{join}}

\newcommand{\pri}{{\mb{pri}}}

\newcommand{\spanc}{span}
\newcommand{\depth}{span}
\newcommand{\TAS}[0]{$\mf{Test-and-Set}$}

\newcommand{\func}[1]{{\sc #1}}
\newcommand{\para}[1]{{\bf \emph{#1}}\,}
\newcommand{\intersection}{\func{Intersection}}
\newcommand{\union}{\func{Union}}
\newcommand{\difference}{\func{Difference}}

\newcommand{\intersect}{\func{Intersect}}
\newcommand{\diff}{\func{Diff.}}

%% file: abstract.tex
\begin{abstract}
  In this paper we develop optimal algorithms in the binary-forking model for a variety of fundamental problems, including sorting, semisorting, list ranking, tree contraction, range minima, and ordered set union, intersection and difference.  In the binary-forking model, tasks can only fork into two child tasks, but can do so recursively and asynchronously.  The tasks share memory, supporting reads, writes and test-and-sets.  Costs are measured in terms of work (total number of instructions), and span (longest dependence chain).

  The binary-forking model is meant to capture both algorithm performance and algorithm-design considerations on many existing multithreaded languages, which are also asynchronous and rely on binary forks either explicitly or under the covers. In contrast to the widely studied PRAM model, it does not assume arbitrary-way forks nor synchronous operations, both of which are hard to implement in modern hardware.  While optimal PRAM algorithms are known for the problems studied herein, it turns out that arbitrary-way forking and strict synchronization are powerful, if unrealistic, capabilities. Natural simulations of these PRAM algorithms in the binary-forking model (i.e., implementations in existing parallel languages) incur an $\Omega(\log n)$ overhead in span. This paper explores techniques for designing optimal algorithms when limited to binary forking and assuming asynchrony. All algorithms described in this paper are the first algorithms with optimal work and span in the binary-forking model.  Most of the algorithms are simple.  Many are randomized.

\end{abstract}

\hide{Due to the asynchronous nature of the model, and a variety
  schedulers that are efficient in both theory and practice, variants
  of the model are widely used in practice in languages such as Cilk
  and Java Fork-Join.  PRAM algorithms can be simulated in the model
  but at a loss of a factor of $\Omega(\log n)$ so most PRAM
  algorithms are not optimal in the model even if optimal on the PRAM.
  All algorithms we describe are optimal in work and span (logarithmic
  in span).  Several are randomized.  Beyond being the first optimal
  algorithms for their problems in the model, most are very simple.}
  %% Parallel machines are ubiquitous nowadays, but they are highly asynchronous in processor rates, and the number of processors available can also change over time.
  %% Because of these reasons, as compared to the PRAM, this model with limited forking and dynamic parallelism better captures properties of today's parallel shared-memory platforms, and is widely-used in parallel algorithm research in recent years.
  %% Therefore, it is natural to ask, \emph{how does the restriction of binary forking change theoretical results for a variety of parallel algorithms?}
  %% Perhaps surprisingly, the answer to this question is not even known for basic problems such as sorting, and list/tree contraction.
  %% In this paper, we show very simple (randomized) parallel algorithms for fundamental algorithmic building blocks such as sorting, semisorting, list/tree contraction, random permutation, ordered-set operations, and range minimum queries.
  %% All these new algorithms have optimal work and span bounds in the binary-forking model.
  %% Meanwhile, most of the new algorithms are arguably simpler than the classic PRAM algorithms, so these algorithms are of interests even without considering the restriction of binary-forking.

%% file: intro.tex
\section{Introduction}

In this paper we present several results on the \bfmodel{}.  The model assumes a collection of threads that can be created dynamically and can run asynchronously in parallel.  Each thread acts like a standard random-access machine (RAM), with a constant number of shared registers and sharing a common main memory.  The model includes a \texttt{fork} instruction that forks an asynchronous child thread.  A computation starts with a single thread and finishes when all threads end.  In addition to reads and writes to the shared memory, the model includes a test-and-set (\tas) instruction.
Costs are measured in terms of the work (total number of instructions executed among all threads) and the span (the longest sequence of dependent instructions).

The \bfmodel{} is meant to capture the performance of algorithms on modern multicore shared-memory machines.  Variants of the model have been widely studied~\cite{CLRS,agrawal2014batching,Acar02,blelloch2010low,BCGRCK08,BL98,blumofe1999scheduling,BG04,Blelloch1998,blelloch1999pipelining,BlellochFiGi11,BST12,Cole17,CRSB13,tang2015cache,dinh2016extending,chowdhury2017provably}.  They are also widely used in practice, and supported by programming systems such as Cilk~\cite{frigo1998implementation}, the Java fork-join framework~\cite{Java-fork-join}, X10~\cite{charles2005x10}, Habanero~\cite{budimlic2011design}, Intel Threading Building Blocks (TBB)~\cite{TBB}, and the Microsoft Task Parallel Library~\cite{TPL}.

The binary forking model and variants are practical on multicore shared-memory machines in part because they are mostly asynchronous, and in part due to the dynamic binary forking.  Asynchrony is important because the processors (cores) on modern machines are themselves highly asynchronous, due to varying delays from cache misses, processor pipelines, branch prediction, hyper-threading, changing clock speeds, interrupts, the operating system scheduler, and several other factors.
% Reducing synchronization is also one of the motivations for the MPC model~\cite{karloff2010mapreduce,goodrich2011sorting}, which is designed for distributed systems.
Binary forking is important since it allows for efficient scheduling in both theory and practice, especially in the asynchronous setting~\cite{blumofe1999scheduling,BGM99,ABP01,ACGRS18}.  Efficient scheduling can be achieved even when the number of available processors changes over time~\cite{ABP01}, which often happens in practice due to shared resources, background jobs, or failed processors.

Due to these considerations, it would seem that these models are more practical for designing parallel algorithms than the more traditional PRAM model~\cite{SV81}, which assumes strict synchronization on each step, and a fixed number of processors.  One can also argue that they are a more convenient model for designing parallel algorithms, allowing, for example, the easy design of parallel divide-and-conquer algorithms, and avoiding the need to schedule by hand~\cite{Blelloch96}.  The PRAM can be simulated on the \bfmodel{} by forking $P$ threads in a tree for each step of the PRAM.  However, this has a $O(\log n)$ overhead in span.  This means that algorithms that are optimal on the PRAM are not necessarily optimal when mapped to the \bfmodel{}.  For example, Cole's ingenious pipelined merge sort on $n$ keys and processors takes optimal $O(\log n)$ parallel time (span) on the PRAM~\cite{Cole1988}, but requires $O(\log^2 n)$ span in the \bfmodel{} due to the cost of synchronization.  On the other hand a $O(n \log n)$ work and $O(\log n \log \log n)$ span algorithms in the \bfmodel{} is know~\cite{Cole17}.  Therefore finding more efficient direct algorithms for the \bfmodel{} is an interesting problem.  Known results are outlined in Section~\ref{sec:related}.

The variants of the \bfmodel{} differ in how they synchronize.  The most common variant is binary fork-join model where every fork corresponds to a later join, and the fork and corresponding joins are properly nested~\cite{CLRS,Acar02,blelloch2010low,BL98,blumofe1999scheduling,Blelloch1998,BlellochFiGi11,Cole17}.  Other models allow more powerful synchronization primitives ~\cite{blelloch1999pipelining,CRSB13,tang2015cache,dinh2016extending,chowdhury2017provably,ABP01}.  In this paper we allow a test-and-set (\tas{}), which is a memory operation that atomically checks if a memory location is zero, returning the result, and sets it to one.  This seems to give some power over the pure fork-join model.   We make use of the \tas{} in many of our algorithms.  We justify including a
\tas{} instruction by noting that all modern multicore hardware includes the instruction.  Furthermore all existing theoretical and practical implementations of the fork-join model require the test-and-set, or equivalently powerful operation to implement the join.

In this paper we describe several algorithms for fundamental problems that are optimal in both work and span in the \bfmodel.
In particular, we show the following results.

\begin{theorem}[Main Theorem]
  Sorting, semisorting, list/tree contraction, random permutation, ordered-set operations, and range minimum queries can be computed in the \bfmodel{} with optimal work and \depth{} ($O(\log n)$).  In many cases the algorithms are randomized, as summarized in Table~\ref{tab:summary}.
\end{theorem}

\input{introtable}

To achieve these bounds, we develop interesting algorithmic approaches.
% since previous algorithms based on PRAM rely on arbitrary-way forks and do not have optimal span on the new model.
For some of them, we are inspired by recent results on identifying dependences in sequential iterative algorithms~\cite{blelloch2012internally,BGSS16,shun2015sequential}.
This paper discusses a non-trivial approach to convert the dependence DAG into an algorithm in the \bfmodel{} while maintaining the span of the algorithm to be the same as the longest chain in the DAG.
This leads to particularly simple algorithms, even compared to
previous PRAM algorithms whose span is suboptimal when translated to the \bfmodel{}.
For some other algorithms, we use the $n^\epsilon$-way divide-and-conquer scheme.
% on functional data structures.
By splitting the problem into $n^\epsilon$ sub-problems and solving them in parallel in logarithmic time, we are able to achieve $O(\log n)$ span for the original problem.
Our results on ordered sets are the best known (optimal work in the comparison model and $O(\log n)$ span) even when translated to other models such as the PRAM.
% \guy{Not sure we still define race-free.}\yihan{I think we are trying to say the algorithms are race-free in other models. Maybe we can instead say a CREW PRAM?}

We note that for many of the problems we describe, it remains open whether the same bounds can be achieved deterministically, and
also whether they can be achieved in the binary fork-join model without a \tas{}.   One could argue that avoiding a \tas{} is more elegant.

\subsection{Related Work}
\label{sec:related}

There have been many existing parallel algorithms designed based on variants of the \bfmodel{} (e.g.,~\cite{agrawal2014batching,Acar02,blelloch2010low,BCGRCK08,BG04,Blelloch1998,blelloch1999pipelining,BlellochFiGi11,BST12,Cole17,CRSB13,BBFGGMS16,tang2015cache,dinh2016extending,chowdhury2017provably,BGSS18,dhulipala2020semi,BBFGGMS18,Dhulipala2018,blelloch2020randomized}).  Many of the results are in the setting of cache-efficient algorithms.
This is because binary forking in conjunction with work-stealing or space-bounded schedulers leads to strong bounds on the number of cache misses on multiprocessors with various cache configurations~\cite{Acar02,CRSB13,Cole17,BlellochFiGi11}.

%\guy{Do we define whp?}

In the binary fork-join model, Blelloch et al.~\cite{blelloch2010low} give work-efficient $O(\log n)$ span algorithms for prefix sums and merging, and a work-efficient randomized sorting algorithm with $O(\log^{3/2} n)$ span \whp{}\footnote{We use the term $O(f(n))$ with high probability (\whp) in $n$ to indicate the bound $O(kf(n))$ holds with probability at least $1-1/n^k$ for any $k\ge 1$. With clear context we drop ``in $n$''.}.  Cole and Ramachandran~\cite{Cole17} improved this and gave a deterministic algorithm with span $O(\log n \log \log n)$.  This is currently the best known result for sorting in the binary fork-join model, without a test-and-set, and also for deterministic sorting even with a test-and-set.

%\footnote{We use the term $O(f(n))$ with high probability (\whp) in $n$ to indicate the bound $O(kf(n))$ holds with probability at least $1−1/n^k$ for any $k \ge 1$.With clear context we drop ``in $n$''}

Allowing for more powerful synchronization,
Blelloch et al.~\cite{Blelloch1998,blelloch1999pipelining} discussed how to implement futures using the \tas{} instruction, which leads to some low-span binary-forking algorithms Tang et al.~\cite{tang2015cache,dinh2016extending,chowdhury2017provably} described some dynamic programming algorithms, in the setting of cache efficiency.  They also use a \tas{} for synchronization.    With this they can reduce the span of a variety of algorithms over fork-join computations without the atomic synchronizations.  Without considering the additional support for cache efficiency, we believe their model is equivalent to the \bfmodel{}.% with \tas{}.

%% file: introtable.tex
\begin{table}[t!]\small
%\captionsetup{justification=centering}
\begin{center}
    %\begin{tabular}{@{ }l@{  }l@{  }c@{  }l@{}}
    \begin{tabular}{lccl}
    \hline
    Problem & & Work & \multicolumn{1}{c}{Span}\\
    \hline%\hline
    \textbf{List Contraction} & \multicolumn{1}{c}{Sec \ref{sec:list}}
                                                           & $O(n)$ &
                                                                      $O(\log
                                                                      n)^*$
                                \\
    \textbf{Sorting} & \multicolumn{1}{c}{Sec \ref{sec:sorting}} & $O(n\log n)^\dagger$ & $O(\log n)^*$\\
    \textbf{Semisorting} & \multicolumn{1}{c}{Sec \ref{sec:sorting}} & $O(n)^\dagger$ & $O(\log n)^*$ \\
    \textbf{Random Permutation} & \multicolumn{1}{c}{Sec \ref{sec:permutation}} & $O(n)^\dagger$ & $O(\log n)^*$  \\
    \textbf{Range Minimum Query} & \multicolumn{1}{c}{Sec \ref{sec:rmq}} & $O(n)$ & $O(\log n)$  \\
    \textbf{Tree Contraction} & \multicolumn{1}{c}{Sec \ref{sec:tree}} & $O(n)$ & $O(\log n)^*$   \\
    \textbf{Ordered-Set Operations} & \multirow{2}{*}{Sec \ref{sec:setset}} & \multirow{2}[1]{*}{$O(m\log(\frac{n}{m}+1))$} & \multirow{2}[1]{*}{$O(\log n)$ } \\
    \textbf{(\union{}, \intersect{}, \diff{})} &       &       &  \\
    \hline
    \end{tabular}%
\end{center}
    \caption{\small The bounds of our new algorithms in the \bfmodel{}.
    For ordered-set operations, $n$ and $m\le n$ are sizes of two sets.
%\texttt{TS} indicates that the
%%algorithm uses the test-and-set instruction.
 $*$: with high probability
      (\whp{}).
\quad $\dagger$: in expectation.
Bounds without superscripts are worst-case bounds.
The ordered-set algorithms can work in \bfmodel{} only with \join{} supported (either by using
\tas{} or just as a default primitive), and the
rest make use of \tas{}.
%\guy{Is this right?}\yihan{reworded}
}
\label{tab:summary}
    \vspace{-2.5em}
\end{table}%
% Regarding optimal parallel algorithms, the classic algorithms for reduce, merge, and prefix sum~\cite{JaJa92,Blelloch89} can easily
% be converted to be optimal in the \bfmodel~\cite{blelloch2010low}.
% Matrix multiplication can be viewed as independent inner products, which is also optimal.
% However, for many fundamental problems discussed in this paper, we are unaware of existing parallel algorithms with optimal work and $O(\log n)$ \depth{}.
% The classic algorithms use synchronization in a crucial way to reduce the \depth{}, and we see no way to easily remove it.

%% file: nmodels.tex
\section{Models and Simulations}
\label{sec:mtram}

\newcommand{\tram}{TRAM}

Here we describe the \bfmodel{} and its relationship to more
traditional models of parallel computing, including the PRAM and
circuit models.  The \bfmodel{} falls into the class of multithreaded
models~\cite{blumofe1999scheduling,BGM99,ABP01,CR17b,BL98}.
Multithreaded computational models assume a collection of threads
(sometimes called processes or tasks) that can be dynamically created,
and generally run asynchronously.  Cost is determined in terms of the
total work and the computational span (also called depth or critical
path length).  There are several variants on multithreaded models depending on
how many threads can be forked, how they synchronize, and assumptions
about how the memory can be accessed.
To be concrete, we define a specific model in this
paper.  % based on the models described in the references above.

\paragraph{\textbf{The \bfmodel.}}
The \defn{\bfmodel} consists of \defn{\thread{}s} that share a common
memory.  Each \thread{} acts like a sequential RAM---it works on a
program stored in the shared memory, has a constant number of
registers (including a program counter), and has standard RAM
instructions (including an \insend{} instruction to finish the
computation).  The \bfmodel{} extends the RAM with a \forkins{}
instruction, which forks a \emph{child} thread.  We also employ a
special \insend{} instruction named \finish{} to indicate the completion of the whole computation.
The \forkins{}
instruction sets the first register to zero in the parent (\emph{forking})
thread and to one in the child (\emph{forked}) thread, to distinguish them.
Otherwise the states of the threads are identical, including the
program counter to the next instruction.
As is standard with the sequential RAM~\cite{Tarjan83}, we
assume that for input size $n$, all memory locations and
registers can hold $O(\log n)$ bits.

In addition to reads and writes, we include a \texttt{test-and-set}
(\tas) instruction in the \bfmodel{} for accessing memory.  The \tas{}
is an atomic instruction that reads a memory location and if the
memory location is \tszero{}, sets it to \tsone{}, returning
\tszero{}. Otherwise it leaves the value unchanged returning \tsone{}.
We note that all currently produced processors support the \tas{}
instruction in hardware. %, as well as more powerful memory instructions. \yihan{I don't quite understand the last sentence here?}

In a binary-forking model, a computation starts with a
single \defn{initial} \thread{} and finishes when \finish{} is called.
The invocation to an \finish{} can be determined by the algorithm, for example,
through using \tas{} instructions (e.g., to implement \joinins{} instructions, see below).
A computation in the \bfmodel{} can therefore be viewed as a tree where each
node is an instruction with the next instruction as a child, and where
the \texttt{fork} instruction has two children corresponding to the next
instruction of the original forking thread and the first instruction of
the forked thread.  The root of the tree is the first instruction of
the initial thread.  We define the \emph{work} of a computation as the
size of the tree (total number of instructions) and the
\emph{\spanc{}} as the depth of the tree (longest path of
instructions).  We assume the results of memory operations are
consistent with some total order (linearization) of the instructions
that preserves the partial order defined by the tree.  For example, a
read will return the value of the previous write or \tas{} to the same
location in the total order.  The choice of total order can affect the
results of a program since \thread{}s can communicate through the
shared memory.  In general, therefore, computations are
nondeterministic.

To simplify issues of parallel memory allocation we assume there is an
\allocateins{} instruction that takes a positive integer $n$ and
allocates a contiguous block of $n$ memory locations, returning a
pointer to the block, and a \freeins{} instruction that given a
pointer to an allocated block, frees it. %\guy{Added this paragraph.}

We use $\mathcal{BF}(W(n),S(n))$ to denote the class of algorithms that require
$O(W(n))$ work and $O(S(n))$ span for inputs of size $n$ in the \bfmodel{}.
% Similarly we use $\TRAM(W(n),S(n))$ for the \tram{}
% model with arbitrary forking.
We use
$\mathcal{BF}^k$ % ($\TRAM^k$)
when $S(n) = O(\log^k(n))$ and $W(n)$ is
polynomial in $n$, and
$\mathcal{BF}^*$ % ($\TRAM^*$)
when the span is polylogarithmic and the
work is polynomial.

The \bfmodel{} can be extended to support arbitrary-way forking
instead of binary.  In particular, the \forkins{}
instruction can take an integer specifying the number of threads to
fork, and each forked thread then gets a unique integer identifier in
a register.  The focus of this paper, however, is on binary forking
since there are no known optimal scheduling results for arbitrary-way
forking (see below).  The model can also be augmented with more
powerful atomic memory operation.  For instance, some
algorithms~\cite{agrawal2014batching,BBFGGMS16,dhulipala2020semi,BBFGGMS18,Dhulipala2018}
use compare-and-swap (\cas{}) in addition to the above-mentioned
model.   We refer to this model as the \bfmodel{} with
\cas{}.  A \tas{} is sufficient for our algorithms.  %\guy{Moved CAS
  %discussion here.}

\paragraph{\textbf{Joining}}
It can be useful to join threads after forking them, and many models
support such joining~\cite{blumofe1999scheduling,BGM99,ABP01,CR17b}.  This can be implemented by
adding a \texttt{join} instruction to the \bfmodel{}.  When
reaching a \texttt{join} instruction in thread $t$ the \emph{forking thread} $t$ must ``wait''
until its most recently \emph{forked} child thread $t_c$ ends.  Specifically,
in the partial order of the tree mentioned above, it means the partial
order is augmented with a dependence from the \insend{} instruction of $t_c$
to the \texttt{join} instruction of $t$.  This partial order is now a
series-parallel DAG instead of a tree, and the total order has to be
consistent with it.  As before, the work is the total number of
instructions, but now the \spanc{} is the longest path of instructions
in the DAG instead of tree.  We call this the \emph{\fjmodel{}}.

Joining can easily be implemented in the \bfmodel{} without a built-in
\joinins{} instruction, but by using the \tas{} instruction.  To implement a \joinins{},
before each fork we initialize a ``synchronization'' location to
\tszero.  For the forking and the forked threads, whichever finishes
later is responsible for processing the rest of the computation after the \join{}.
This is determined by reaching consensus through the synchronization location.
When
the forking thread $T$ reaches a \joinins{} it saves its registers and
then performs a \tas{} on the corresponding synchronization location.
If the \tas{} returns \tsone{}, this means that the other thread has
already finished and set it to \tsone{} first, and $T$ can continue to the next instruction in the program.
Otherwise, it means that the other thread has not finished yet, and thus $T$ ends because the other thread will take over the rest of the computation later.
When the forked thread
reaches its \texttt{end}, it also performs a \tas{} on the synchronization location. Similarly, if the \tas{} returns \tszero{} it ends, otherwise it loads the registers saved by the forking thread, and jumps to the stored program counter.
% This ensures the continuation of computation after the join point
% waits until both the parent thread has reacheded the join, and that
% the child has ended, but does not waste work waiting.
This
implementation preserves work and span within a constant factor.
By using \forkins{} and \joinins{} one can also simulate a regular parallel for-loop of size $n$
using divide-and-conquer, which takes $\Theta(\log n)$ \spanc{} to fork and synchronize.
%Instead of a $O(1)$ \spanc{} as in the standard \fjmodel{} allowing arbitrary-way forking and joining,
%this parallel for-loop in \bfmodel{} would take $\Theta(\log n)$ \spanc{} to fork and synchronize.

The simulation implies that the \bfmodel{} is as least as powerful as
the \fjmodel{} (with or without \tas{}).  We note that unlike \fjmodel{},
by using a general \tas{} instead of just a \join{}, the parallelism
supported by \bfmodel{} is not necessary nested.
We point out that to
implement a constant-time join seems to require an operation at least as powerful as
\tas{}.  In particular reads and writes by
themselves are not powerful enough to get consensus among even just
two processes in a wait-free manner, and \tas{} is the seems to be the
least powerful
memory operation that can achieve two process consensus~\cite{Herlihy91}.
This suggests a primitive as powerful as \tas{} is necessary to
efficiently implement a join on an asynchronous machine since the two
joining threads need to agree (reach consensus) on who will run the
continuation.

% For these reasons we treat the \bfmodel{} with \tas{} as the
% primitive model and the \fjmodel{} as derived.
% Meanwhile, we note that algorithms based on \fjmodel{} can be stronger since they are deterministic and more aesthetic.
%In this paper, our algorithms for ordered sets in Section~\ref{sec:setset} make use of the \joinins{} instructions which implies implicit uses of \tas{}, while other algorithms explicitly use \tas{}.

\paragraph{\textbf{PRAM}}

For background, we give a brief description of the PRAM
model~\cite{SV81}.  A PRAM consists of $p$ processors, each a sequential
random access machine (RAM), connected to a common shared memory of
unbounded size.  Processors run synchronously in lockstep.  Although
processors have their own instruction pointer, in typical algorithms
they all run the same program.  There are several variants of the
model depending on how concurrent accesses to shared memory are
handled---e.g., CRCW allows concurrent reads and writes, and EREW
requires exclusive reads and writes.  For concurrent writes, in this
paper we assume an arbitrary element is written (the most standard
assumption).  A more detailed description of the model and its
variants can be found in J\'{a}J\'a's book on parallel
algorithms~\cite{JaJa92}.  As with the \bfmodel{}, we assume that for an
input of size $n$, memory locations and registers contain at most
$O(\log n)$ bits.   We use
$\PRAM(W(n),S(n))$ to indicate PRAM algorithms that run in $O(W(n))$
work (processor-time product) and $S(n)$ time,
$\PRAM^k$ when the time is $O(\log^k n)$, and
$\PRAM^*$ when it is polylogarithmic (both with polynomial work).

\paragraph{\textbf{Relationship to the PRAM}}

There have been many scheduling results showing how to schedule
binary and multiway forking on various machine
models~\cite{blumofe1999scheduling,BGM99,ABP01}.  For example, the
following theorem can bound the runtime for programs
in the \bfmodel{} on a PRAM.
%$P$-processor loosely-synchronized parallel machine.

\begin{theorem}[\cite{BL98,ABP01}]
  \label{thm:workstealing}
  Any computation in the \bfmodel{} that does $W$ work and has $S$
  span can be simulated on $P$ processors of a loosely synchronous
  parallel machine or the CRCW PRAM in \[O\left(\frac{W}{P}+S\right)\]
  time \whp{} in $W$.
\end{theorem}%\yan{W/P outside of big-O?}

\noindent
This is asymptotically optimal (modulo randomization) since the
simulation must require the maximum of $W/P$ (assuming perfect balance
of work) and $S$ (assuming perfect progress along the critical path).
The result is based on a work-stealing scheduler.  A slight variant of
the theorem applies in a more general setting where individual
processors can stop and start~\cite{ABP01} and $P$ is the average
number of processors available.
% The best that is known for deterministic dynamic scheduling is
% \[O\left(\frac{W}{P}+S \log^{*} p \right).\] This is based on using
% approximate prefix sums~\cite{Goodrich94,BGM99} and works not just for
% both binary and arbitrary forking, but is suboptimal.

Importantly, in the other direction, simulating a $p$-processor PRAM,
even the weakest EREW PRAMg requires a
$\Theta(\log p)$ factor loss in span on the \bfmodel{}.  This is a lower bound for any
simulation that is faithful to the synchronous steps since just
forking $p$ parallel instructions (one step on a PRAM) requires at
least $\log p$ steps on the \bfmodel.

\paragraph{\textbf{Relationship to Circuit Models}}
Beyond the PRAM we can ask about the relationship to circuit models and to
bounded space.  Here we use $\NC$ for Nick's class, $\AC$ when allowing
unbounded in-degree, and $\LogSpace$ for logspace~\cite{Borodin77,KarpR90}.  We first
note that $\NC = \BF^*$.  This follows directly from the
PRAM simulations since
$\NC = \PRAM^*$~\cite{KarpR90}.    We also have the following
more fine-grained results.
\ifx\fullversion\undefined We show the proof in the full version of this paper \cite{blelloch2019optimal}.
\else
We show the proof in Appendix \ref{app:complexity}.
\fi
\begin{theorem}
\label{thm:complexity}
\[ \NC^1 \subseteq \LogSpace \subseteq \BF^1 \subseteq \AC^1 = \PRAM_{CRCW}^1 \subseteq \NC^2 \]
\end{theorem}

%% file: listcontraction.tex
\section{List Contraction}\label{sec:list}
List
ranking~\cite{Wyllie79,Vishkin84,Cole85,AndersonMiller1990,KarpR90,JaJa92,Baase93,Vishkin93,Reid-Miller93,Randade98,jacob2014complexity}
is one of the canonical problems in the study of parallel algorithms.
The problem is: given a set of linked lists, compute for each element
its position in the list to which it belongs.
The problem can be solved by \emph{list contraction},
which contracts a list by following the pointers in the list.
After contraction one can rank the list by a second phase that expands it back out.
%This technique is especially useful in the parallel setting since multiple nodes in the list
%can perform contraction in parallel.
The problem has
received considerable attention because of: (1) its fundamental nature
as a pointer-based algorithm that seems on the surface to be
sequential; and (2) it has many applications as a subroutine in other
algorithms.  Wyllie~\cite{Wyllie79} first gave an $O(n \log n)$ work
and $O(\log n)$ time algorithm for the problem on the PRAM over 40
years ago.  This was later improved to a linear work
algorithm~\cite{Cole1988ApproximatePS}.  Although this problem has
been extensively studied, to the best of our knowledge, all existing
linear-work algorithms have ${\Omega}(\log^2 n)$ \spanc{} in the
\bfmodel{} because they are all round-based algorithms and run in
${\Omega}(\log n)$ rounds.  The main result of this section is a
randomized, linear work, logarithmic span algorithm in the \bfmodel{}.
Then we also describe how to adapt Wyllie's
algorithm to the \bfmodel{} to achieve $O(n \log n)$ work and
$O(\log n)$ \spanc{}; while not work optimal, this latter algorithm is
deterministic.  Both algorithms are the first in the \bfmodel{} to
achieve $O(\log n)$ span.
%\guy{why tilde?  I cannot find any reference to such an algorithm.}

We now present a simple randomized algorithm (Algorithm~\ref{algo:lc}) for list contraction
that is theoretically optimal (linear work, and $O(\log n)$ \spanc{}~\whp{}) in
the \bfmodel{}.
This algorithm is inspired by the list contraction algorithm
in~\cite{shun2015sequential}, but it improves the \spanc{} by $\Theta(\log n)$, and is
quite simple.

The main challenge in designing a work-efficient parallel list
contraction algorithm is to avoid simultaneously trying to splice-out
two consecutive elements.  One solution is via assigning each element
a priority from a random permutation.  An element can be spliced out
only when it has a smaller priority than its previous and next
elements, so the neighbor elements cannot be spliced out
simultaneously.  If the splicing is executed in rounds (namely,
splicing out all possible elements in a round-based manner), Shun
et al.~\cite{shun2015sequential} show that the entire algorithm
requires $\Theta(\log n)$ rounds \whp{}, leading to $\Theta(\log^2 n)$
\spanc{} \whp{} in the \bfmodel{}.  The dependence structure of the
computation is equivalent to a randomized binary tree.  On each round
we can remove all leaf nodes so the full tree is processed in a number
of rounds proportional to the tree depth.  An example is illustrated
in Figure~\ref{fig:list}.

\input{figure-list}
\input{code-list}

After a more careful investigation, we note that the splicing can
proceed asynchronously, and not necessarily based on rounds.  For
example, the last spliced node with priority 7 separates the list into two disjoint
sublists, and the contractions on the two sides are independent and can run
asynchronously.
Conceptually we can do this recursively, and the recursion depth is $\Theta(\log
n)$~\whp{}~\cite{shun2015sequential}.  Unfortunately, we cannot directly
apply the divide-and-conquer approach since $L$ is stored as a linked
list and deciding the elements within sublists is as
hard as the list contraction algorithm itself.

We present our algorithm in Algorithm~\ref{algo:lc}.
Starting from the leaves, Algorithm~\ref{algo:lc} performs equivalent steps to the
algorithm in~\cite{shun2015sequential}, but its \spanc{} is
$\Theta(\log n)$ in the \bfmodel{}; this improvement is achieved by
allowing the splicing in each round to run asynchronously.
The key idea is that, instead
of checking all element for readiness in each round, as long as two children of a node $c$ finished contracting,
we trigger $c$ to start contracting immediately. The child of $c$ that finished later is responsible
for take over $c$, and thus can start immediately.
In particular, in the algorithm, a
parallel-for loop (Line \ref{line:listfor}) generates $n$ tasks (threads) each for a node in the list.
The loop can be
implemented by binary forking for $\log_2 n$ levels.
Only leaf nodes start the execution, and non-leaf nodes
quit immediately (Line \ref{line:checkleave}.
These leaves will splice themselves out (Line~\ref{line:splice}), and then try to move upward and splice its parent (Line \ref{line:grabparent}).
We note that a node $c$ cannot be contracted until both of its children have been spliced out.
%$c$ can be spliced out by either of its child, but only one child should take over $c$.
Thus we make the child of $c$ that finishes its splicing later to take over $c$.
This is achieved by letting the two children compete through \TAS{} the \emph{flag} field in $c$ (Line \ref{line:listbreak}).
Whichever arrives later takes over and contracts the parent $c'$ (Line \ref{line:cparent}), and the first one simply terminates its computation (Line \ref{line:listbreak}) and let the second one to take continuation. As an example in Figure \ref{fig:list}, the threads for nodes~1 and~2 will
both try to work on node~5 after they finish their first splicings.
They will both attempt to \TAS{} the $\mb{flag}$ of node~5. The one coming first succeeds and terminates, and the later
one will fail and continue splicing node~5.
We initialize the $\mb{flag}$ for each node to be 0, except for those with 0 or 1 child (Line \ref{line:initbegin}), for which we set $\mb{flag}$ directly to 1 (they do not need to wait for two children).

%It is
%possible that two tasks grasp the same element from both sides.

\begin{theorem}
  Algorithm~\ref{algo:lc} for list contraction does $O(n)$ work (worst
  case) and has $O(\log n)$ \spanc{} \whp{} in $n$ in the \bfmodel{}. % with \texttt{TS}.
\end{theorem}
\begin{proof}%[Proof sketch]
The correctness of this algorithm can be shown as it applies the same
operations as the list contraction algorithm
in~\cite{shun2015sequential}, although Algorithm~\ref{algo:lc} runs in a much less synchronous manner.
The execution of each thread corresponds to a tree-path in the dependence structure starting from a leaf node and ending on either the root or when winning a \TAS{}.
A possible decomposition of the example is shown in the caption of Figure~\ref{fig:list}.
This observation also indicates that the number of
iterations of the while-loop on line~\ref{line:tscheck} for any task is $O(\log
n)$ \whp{}, bounded by the tree height.
The \spanc{} is therefore $O(\log n)$ \whp{}.
The work is linear because every time
Line~\ref{line:splice}--\ref{line:listbreak} is executed, one element will be spliced out.
\end{proof}

It is worth noting that, even disregarding the improved \depth{} for
the \bfmodel{}, we believe this algorithm is conceptually simpler and
easier to implement compared to existing linear-work, logarithmic-time
PRAM algorithms~\cite{cole1989faster,AndersonMiller1990}.  Our algorithm requires starting with a random permutation (discussed further in Section~\ref{sec:permutation}).
We note that it is
straightforward to extend the analysis to using integer
priorities instead of random permutations, where
the integer priorities are chosen independently and uniformly from the range $[1,n^k]$, for
$k\ge 2$, with ties broken arbitrarily.

%But it is
%straightforward to extend the analysis when instead of choosing integer
%priorities independently and uniformly from the range $[1,n^k]$, for
%$k\ge 2$, with ties broken arbitrarily.

\paragraph{Binary Forking Wyllie.}
Here we outline a binary forking version of Wyllie's algorithm with
$O(n \log n)$ work and $O(\log n)$ span, both in the worst
case.  It is useful for our simulation of logspace in $\BF^1$.  The
idea is to allocate an array of $\log_2 n$ cells per node of the list,
each containing two pointers and a boolean value used for \tas{}.
At the end of the algorithm the two pointers in the $i$-th cell
(level) will point to the element $2^i$ links forward and $2^i$
backward in the list (or a null pointer if fewer
than $2^i$ before or after).  The algorithm initially forks off a
thread for each node at level 1 in the list.  A thread is responsible
for splicing out its link at the current level.  It does this by
writing a pointer to the other neighbor to the corresponding pointer
cells of its two neighbors (i.e., splicing itself out), then doing a
\tas{} on the boolean flag of each neighbor.  For each flag on
which it gets a 1 (i.e., it is the second thread to write the pointer
at this level), it forks a thread to splice out that neighbor at the
next level.  Since this fork at the next level only occurs on the
second update to the node, both links at the next level must already
be available.  In general, each splicing step may create 0, 1, or 2
child threads, depending on the timing of arrivals at the neighbors.
The first and last element in each list must start with its flag set
and writes a null pointer to its one neighbor.  As in Wyllie's original
algorithm, it is easy to keep counts to generate the ranks of each
node in a list.  The total work is proportional to the number of
cells, $O(n \log n)$ since each cell gets processed once.  Since each
fork corresponds to performing a splice at a strictly higher level,
the span is proportional to the number of levels, i.e., $O(\log n)$.

%% file: figure-list.tex
\begin{figure}[t]
\vspace{-.5em}
\begin{center}
  \includegraphics[width=.6\columnwidth]{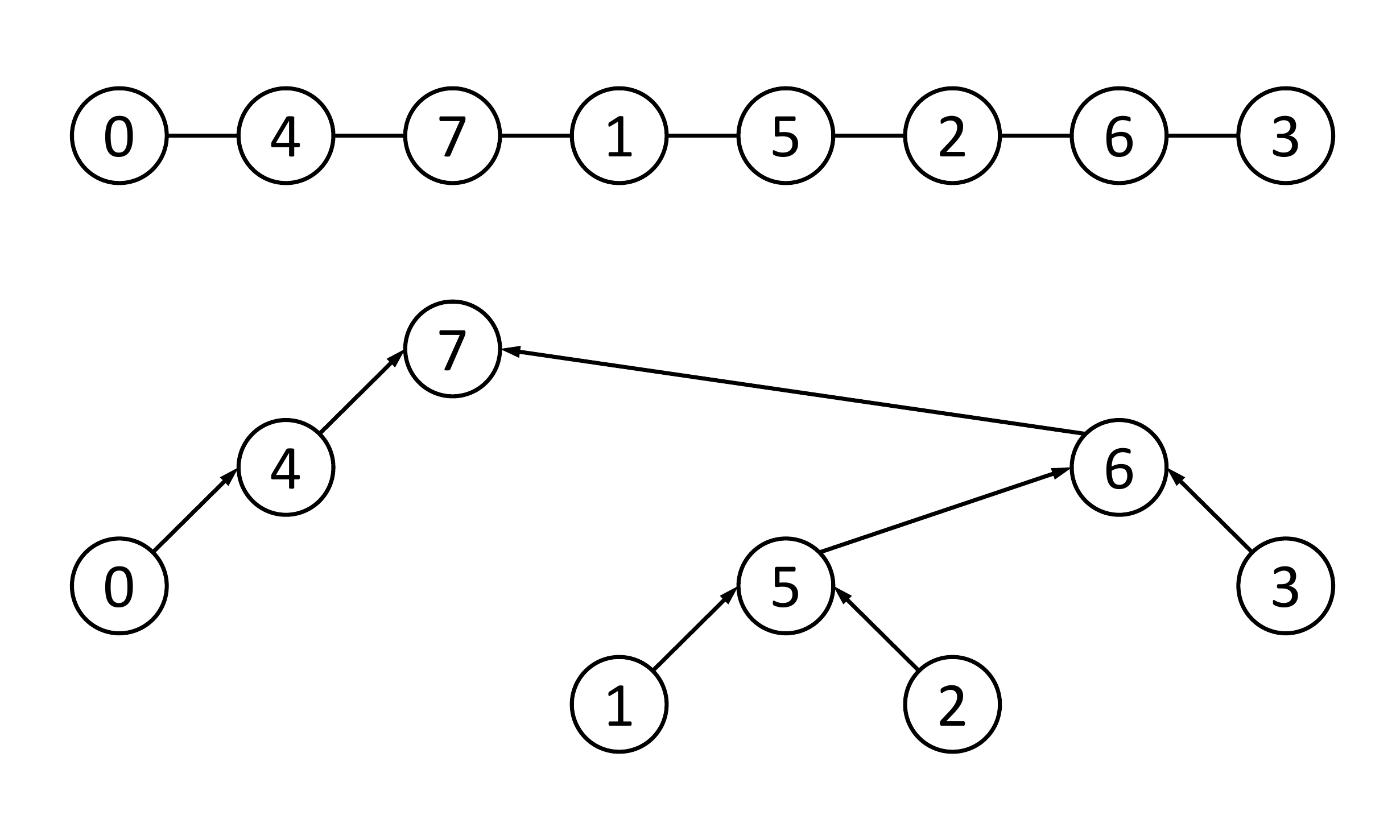}
\end{center}\vspace{-1.5em}
\caption{\small An example of an input list with 8 elements.  The number in each element is the priority drawn from a random permutation.  The dependences of the contractions are shown as a binary tree structure.  In a round-based algorithm~\cite{shun2015sequential}, the execution is in 4 rounds: $\{0,1,2,3\}$, $\{4,5\}$, then $\{6\}$, and finally $\{7\}$.  In Algorithm~\ref{algo:lc}, the execution is asynchronous, and a possible tree-path decomposition is $\{0,4\}$, $\{\varnothing\}$, $\{\varnothing\}$, $\{1\}$, $\{\varnothing\}$, $\{2,5,6,7\}$, $\{\varnothing\}$, and $\{3\}$ for all 8 elements from left to right.  The length of a tree-path is bounded by the tree height.}
\label{fig:list}
\vspace{-1em}
\end{figure}

%% file: code-list.tex
\begin{algorithm}[t]\small
\caption{$\mf{List-Contraction}(L)$}
\label{algo:lc} \small
\SetKwFor{ParForEach}{parallel foreach}{do}{endfch}
\SetKw{Break}{break}
\KwIn{A doubly-linked list $L$ of size $n$.  Each element $l_i$ has a
  random priority ($l_i.p$), next pointer ($l_i.\mb{next}$), previous
  pointer ($l_i.\mb{prev}$) and flag ($l_i.\mb{flag}$).}
  \DontPrintSemicolon
    \vspace{0.5em}
%    $\mb{flag}\gets \{0,\ldots,0\}$\\
    \ParForEach {\upshape element $l_i$ in $L$} {
        \vspace{.2em}\tcp{set flag if zero or one child}
        $l_i.\mb{flag}\gets (\pri(l_i)<\pri(l_i.\mb{prev}))$ or
        $(\pri(l_i)<\pri(l_i.\mb{next}))$ \label{line:initbegin}
    }
    %Set the $\mb{flag}$ for the head and tail elements to be 1\\
    \ParForEach {\upshape element $l_i$ in $L$\label{line:listfor}} {
        $c\gets l_i$\\
        \tcp{Execute only if $c$ is a leaf node}
        \If {(\upshape $(\pri(c)<\pri(c.\mb{prev}))$ \textbf{and}
          $(\pri(c)<\pri(c.\mb{next})) $) \label{line:checkleave}}
          {
          \vspace{.2em}\tcp{Stop when list is contracted into one node}
          \While{\textbf{not} ($c.\mb{prev}$ = \emph{null} \textbf{and} $c.\mb{next}$ = \emph{null})}{
            Splice $c$ out ~\label{line:splice}\\
            Let ${c'}$ be $c.\mb{prev}$ or $c.\mb{next}$ with a smaller priority~\label{line:grabparent}\\
            \tcp{If $c$ is not the last child of $c'$, quit}
            \lIf {\upshape $\mf{Test-and-Set}(c'.\mb{flag})$} {\label{line:tscheck}
                \Break\label{line:listbreak}
            }
            $c\gets c'$~\label{line:cparent}
        }
        }
    }%    \vspace{.5em}
    \vspace{.1in}
\Fn {\upshape$\pri(v)$} {
\lIf*{$v$ = null} {\Return $\infty$} \lElse {\Return $v.p$}
}
\end{algorithm}

%\begin{algorithm}[t]\small
%\caption{$\mf{List-Contraction}(L)$}
%\label{algo:lc} \small
%\SetKwFor{ParForEach}{parallel foreach}{do}{endfch}
%\SetKw{Break}{break}
%\KwIn{A doubly-linked list $L$ of size $n$.  Each element $l_i$ has a
%  random priority ($l_i.p$), next pointer ($l_i.\mb{next}$), previous
%  pointer ($l_i.\mb{prev}$) and flag ($l_i.\mb{flag}$).}
%  \DontPrintSemicolon
%    \vspace{0.5em}
%%    $\mb{flag}\gets \{0,\ldots,0\}$\\
%    \ParForEach {\upshape element $l_i$ in $L$} {
%        \tcp*[h]{flag set if zero or one child}\\
%        $l_i.\mb{flag}\gets (\pri(l_i)<\pri(l_i.\mb{prev}))$ or
%        $(\pri(l_i)<\pri(l_i.\mb{next}))$ \label{line:initbegin}
%    }
%    %Set the $\mb{flag}$ for the head and tail elements to be 1\\
%    \ParForEach {\upshape element $l_i$ in $L$} {
%        $c\gets l_i$\\
%        \While {(\upshape $(\pri(c)<\pri(c.\mb{prev}))$ and
%          $(\pri(c)<\pri(c.\mb{next})) $) \label{line:checkleave}} {
%          \tcp*[h]{$c$ is a leaf node}\\
%          \lIf{$c.\mb{prev}$ = null and $c.\mb{next}$ = null}{\Break}
%            Splice $c$ out. Let ${c'}$ be $c.\mb{prev}$ or $c.\mb{next}$ with a smaller priority~\label{line:splice}\\
%            \If {\upshape not($\mf{Test-and-Set}(c'.\mb{flag})$)} {\label{line:tscheck}
%                \Break\label{line:listbreak} \tcp*[f]{not the last child of $c'$}
%            }
%            $c\gets c'$
%        }
%    }%    \vspace{.5em}
%    \vspace{.1in}
%\Fn {\upshape$\pri(v)$} {
%\lIf*{$v$ = null} {\Return $\infty$} \lElse {\Return $v.p$}
%}
%\end{algorithm}

%% file: sorting.tex
\section{Sorting}
\label{sec:sorting}

In this section we discuss optimal parallel algorithms to comparison
sort and semisort~\cite{Valiant91} $n$ elements using $O(n\log n)$ and
$O(n)$ expected work respectively, and $O(\log n)$ \depth{} \whp{}.
For comparison sort, the best previous work-efficient result in the
binary-forking model requires $O(\log n \log \log n)$
\depth{}~\cite{Cole17}.
% While this algorithm is deterministic,
% simultaneously achieves good cache performance, and does not require
% test and set, it falls short of the optimal $O(\log n)$ span
% bound.
In this paper, we discuss a relatively simple algorithm (Algorithm~\ref{algo:sort}) that sorts $n$ elements in $O(n\log n)$ expected work and $O(\log n)$ \depth{} \whp{}.
%The work bound can be shown with high probability, and details are discussed in the full version of this paper.
%The analysis is also straightforward, so even without considering the \bfmodel{}, it is still an interesting result.

Our algorithm, given in Algorithm~\ref{algo:sort}, is based on sample sorting~\cite{Frazer70}.
It runs recursively.
In the base case when the subproblem size falls below a constant threshold, it sorts sequentially.
Otherwise, for a subproblem of size $n$, the algorithm selects ${n}^{1/3}\log_2 n$ samples uniformly at random, and uses the quadratic-work sorting algorithms to sort these samples (i.e., by making all pairwise comparisons).
These two steps can be done in $o(n)$ work and $O(\log n)$ \depth{} in the \bfmodel{}.
Then the algorithm subselects every $\log_2 n$-th sample to be a pivot, and use these ${n}^{1/3}$ pivots to partition all elements into ${n}^{1/3}+1$ buckets.

\begin{lemma}\label{lem:nc1}
%%For a set $A$ of $n$ elements, when selecting ${n}^{1/3}\log_2 n$ samples selected uniformly at random, and subsampling every $\log_2 n$-th element as $s_1,s_2,...s_{n^{1/3}}$ from a sorted order of these samples, the number of elements in $A$ falls between $s_i$ and $s_{i+1}$ is no more than $c_1rn^{2/3}$ with probability at least $1-n^{-c_1}$ for certain constant $r$ and any constant $c_1>1$.
In the distribution step on Line \ref{line:distribute} in Algorithm \ref{algo:sort}, the number of elements falling into one bucket is no more than $c_1rn^{2/3}$ with probability at least $1-n^{-c_1}$ for certain constant $r$ and any constant $c_1>1$.
\end{lemma}

This follows from Chernoff bound. The algorithm then allocates ${n}^{1/3}+1$ arrays, one per bucket, each with size $2c_1rn^{2/3}$.
Then in parallel, each element uses binary search to decide its corresponding bucket.
It then tries to add itself to a random position in the bucket by using a \tas{} on a flag to reserve it. If the \tas{} fails, it tries again since the slot is already taken.
  %The maximal number of retries is $c_2\log_2 n$ for a constant $c_2>2$.
  We limit the number of retries for each element to be no more than $c_2\log_2 n$.
  If any element cannot find an available slot in this number of retries, the algorithm restarts the process from the random-sampling step (Line~\ref{line:sortsample}).
  Otherwise, after all elements are inserted, the algorithm packs the buckets into contiguous
elements for input to the next recursive calls.

\input{code-sort}

\begin{theorem}\label{thm:sort}
  Algorithm~\ref{algo:sort} sorts $n$ elements in $O(n\log n)$ expected work and $O(\log n)$ \depth{} \whp{} in the \bfmodel{}.
\end{theorem}

To bound span, we need to consider the number of retries and the cost
of each retry along any path to a leaf in the recursion tree.  The
number of retries is upper bounded by a geometric distribution since
each retry is independent, but the probability of that distribution
depends on the level of recursion since problem sizes get smaller.
Furthermore the span of a try also
depends on the level of recursion (it is bounded by $O(\log n_i)$,
where $n_i$ is the input size of level $i$).  To help analyze the span,
we will use the following Lemma.

\begin{lemma}\label{lem:sort}
  Let $X_1\,\cdots,X_m$ be independent geometric random variables, and $X_i$ has success probability $p_i=1-2^{-k^i}$ where $k>1$ is a constant.  Then $\sum_{i=1}^{m}{k^i\cdot X_i}\le O(c k^m)$ holds with probability at least $1-2^{-ck^m}$ for any given constant $c\ge 1$.
\end{lemma}
\begin{proof}
  We view the contribution from each term $k^i\cdot X_i$ to the sum based on a series of independent unbiased coin flips.
  The term $k^i\cdot X_i$ can be considered as the event that we toss $k^i$ coins simultaneously, and if all $k^i$ coins are heads we charge $k^i$ to the sum and this process repeats (corresponding to the geometric random variable with probability $p_i = 1 - 2^{-k^i}$).
  However, in this analysis, we toss one coin at a time until we get a tail, and we charge 1 to the sum for each head before the tail.
  In this way we can only overestimate the sum.
  Hence, $\sum_{i=1}^{m}{k^i\cdot X_i}$ can be upper bounded by the number of heads when tossing an unbiased coin until we see $m$ tails.
  We use Chernoff bound\footnote{\href{https://en.wikipedia.org/wiki/Chernoff_bound}{https://en.wikipedia.org/wiki/Chernoff\_bound}.}
  $\Pr(X\le(1-\delta)\mu)\le e^{-\delta^2\mu/2}$, where $X$ is the sum of indicator random variables, and $\mu=\E[X]$.
  Now let's consider the probability that we see more than $qk^m$ heads before $m$ tails.
  Since $m<k^m$, we analyze the probability to see no more than $k^m$ tails, which only increases the probability.
  In this case, we make $(q+1)k^m$ tosses, so $\mu=(q+1)k^m/2$ and $\delta=(q-1)/(q+1)$.
  The probability is therefore no more than:
%\guy{Sorry, what I thought needed more lines was the mapping of the Chernoff bounds to the first equation, i.e.. where does the $(q-1)/(q+1)$ come from, or the divide by 4, or just about anything in the first equation. What is $X$?   What does $E[X]$ equal?   BTW, isn't this just the tail bound of the binomial distribution?  Some such bounds are given in Wikipedia with references.}
\begin{align*}
  &\exp\left[-{\left(\frac{q-1}{q+1}\right)^2\cdot\frac{(q+1)k^m}{4}}\right] = \exp\left[-\frac{\left(q-1\right)^2k^m}{4(q+1)}\right]\\
  =&\exp\left[{-\frac{(q^2-2q+1)k^m}{4(q+1)}}\right]<\exp\left[{-\frac{(q^2-3q-4)k^m}{4(q+1)}}\right]\\
  =&\exp\left[-\frac{(q+1)(q-4)k^m}{4(q+1)}\right]=e^{-\left(\frac{q}{4}-1\right)k^m}<2^{-\left(\frac{q}{4}-1\right)k^m}
\end{align*}
%\guy{could use a few more steps here.  I dont immediately see how this is derived.}\yihan{added some steps.}
  This proves the lemma by setting $q=4(c+1)$.
\end{proof}

\begin{proof}[Proof of Theorem~\ref{thm:sort}]
  The main challenge is to analyze the work and \depth{} for the distribution cost (Line~\ref{line:distribute}), especially to bound the cost of restarting the distribution step.
  There are two reasons that the call return to Line~\ref{line:sample}: badly chosen pivots such that some buckets contain too many elements and become overfull (defined later), or unlucky random number sequences such that the positions tried by a particular element are all occupied (for more than $c_2\log n$ consecutive slots).
  We say a bucket is overfull if it has more than $c_1rn^{2/3}$ elements (more than half of the allocated space).
  From Lemma \ref{lem:nc1}, the probability of this event is no more than $n^{-c_1}$.
  We pessimistically assume that the distribution step restarts once a bucket is overfull.
  Therefore, for the latter case with a bad random number sequence, the allocated array is always no more than half-full, which is useful in analyzing this case.

  We now analyze the additional costs for the restarts.
  For the latter case, with probability at most $n^{1-c_2}$, at least one element retries more than $c_2\log_2 n$ times.
  For the first case, the probability that any bucket is oversize is $n^{1-c_1}$.
  By setting $c_1$ and $c_2$ to be at least 2, the expected work including restarts is asymptotically bounded by the first round of selecting pivots and distributing the elements.  The work of the first round is bounded by $O(n \log n)$ since there are $n$ elements, each doing a binary search and then each trying at most $O(\log n)$ locations.  Therefore the expected work for each distribution
is $O(n_i \log n_i)$, where $n_i$ is the size of the input.

The total number of elements across any level of recursion is at most $n$ since every element goes to at most one bucket.    Also the size of each input reduces to at most $k n^{2/3}$ from level to level, for
some constant $k$.      The total expected work across each level of recursion therefore decreases geometrically from level to level.      Hence the total work is asymptotically bounded
by the work at the root of the recursion, which is $O(n \log n)$ in expectation.

  We now focus on the \depth{}, and first analyze the case for the chain of subproblems for one element.
  The number of recursive levels is $O(\log \log n)$. For each level with subproblem size $n'$,
  let $c=c_1=c_2\ge 2$. The probability for a restart is less than $2(n')^{1-c}$, and the \depth{} cost for a restart is $c\log_2 n$.
  Treating the number of restarts in each level as a random variable, we can plug in Lemma~\ref{lem:sort} with $k=1.5$ and $m=\log_{1.5}\log_2 n$, and show that the \depth{} of this chain is $O(ck^m)=O(c\log n)$ with probability at least $1-2^{-c\log n}=1-n^{-c}$.
  Then by taking a union bound for the $n$ chains to all leaves of the recursion, the probability is at least $1-n^{1-c}$.
  Combining the analyses of the work and the \depth{} proves the theorem.
\end{proof}

%% From a practical point of view, in Algorithm~\ref{algo:sort} we need to allocate arrays to hold the elements in the buckets.
%% Since the overall sizes for all arrays in a level is the same, we can pre-allocate the entire space at the beginning, and reuse the space in the recursion to avoid frequent memory allocations in all subproblems.

%\guy{probably drop this variant.}
%In Appendix~\ref{app:sort}, we provide an alternative version of this algorithm with more algorithmic details but simpler analysis.

\myparagraph{Semisorting.}
Semisorting reorders an input array of $n$ keys such that equal keys are contiguous but different keys are not necessarily in sorted order.
It can be used to implement integer sort with a small key range.
Semisorting is a widely-used primitive in parallel algorithms (e.g., the random permutation algorithm in Section~\ref{sec:permutation}).

We note that with the new comparison sorting algorithm with optimal work and \depth{}, we can plug it in the semisorting algorithm by Gu et al.~\cite{gu2015top} (Step 3 in Algorithm 1).
The rest of the algorithm is similar to the distribution step but just run for one round, so it naturally fits in the \bfmodel{} with no additional cost.
Hence, this randomized algorithm is optimal in the \bfmodel{}---$O(n)$ expected work and $O(\log n)$ \depth{}~\whp{}.

%% \medskip
%% As a side note, the high probability bounds for both comparison sorting and semisorting regard to the input size $n$.
%% Since sorting is widely used in many parallel algorithms, one should be careful when applying the bounds.

%% file: code-sort.tex
\begin{algorithm}[t]
\caption{$\mf{Comparison-Sort}(A)$}
\label{algo:sort}\small
\SetKwFor{ParForEach}{parallel foreach}{do}{endfch}
\SetKw{Break}{break}
%\KwIn{An unsorted array $A$.}
%\KwOut{An unsorted array $A$.}
  \DontPrintSemicolon
    \vspace{0.5em}
    Let $n=|A|$\\
    \lIf {\upshape $n$ is a constant} {
        Sort the base case and \textbf{return}
    }
    Randomly select ${n}^{1/3}\log_2 n$ samples\label{line:sortsample}\label{line:sample}\\
    Use quadratic sorting algorithm to sort the samples\\
    Subsample ${n}^{1/3}$ pivots from the samples~\label{line:subsample}\\
    Distribute all elements in $A$ to $n^{1/3}+1$ buckets based on the
    samples (to form a partition of $A$ to $A_0, A_1,\ldots,
    A_{n^{1/3}}$).  If failed, restart from Line~\ref{line:sortsample}\label{line:distribute}\\
    \lParForEach {\upshape $i\gets 0$ to $n^{1/3}$} {
        $\mf{Comparison-Sort}(A_i)$
    }%    \vspace{.5em}
\end{algorithm}

%% file: setset/setset.tex
\newcommand{\boundcontent}{m \log \left( \frac{n}{m}+1 \right)}
\newcommand{\bound}{O\left( \boundcontent \right)}
\newcommand{\tchunk}{chunk}
\newcommand{\sketch}{sketch of the tree}

\newcommand{\upper}{upper node}

\newenvironment{proofsketch}{%
  \renewcommand{\proofname}{Proof Sketch}\proof}{\endproof}

\section{Ordered Set-set Operations}
\label{sec:setset}

In this section, we show deterministic algorithms for ordered set-set operations (\union{}, \intersection{} and \difference{}) based on weight-balanced binary search trees. %, with optimal work and span in the \bfmodel{}.
In particular we prove the following theorem.

\begin{theorem} %{\em(Optimal algorithm for ordered set algorithms)}
  \label{thm:setset}
\union{}, \intersection{} and \difference{} of two ordered sets of size $n$ and $m<n$ can be solved in $\bound$ work and $O(\log n)$ span in the \bfmodel{}.    This is optimal for comparison-based algorithms.
\end{theorem}

Our approach is based on a (roughly $\sqrt{n}$-way) divide-and-conquer algorithm with lazy reconstruction-based rebalancing. At a high-level, for two sets of size $n$ and $m\,(\le n)$, we will split both trees with $d-1$ pivots equally distributed among the $m+n$ elements, where $d=\Theta(\sqrt{m+n})$ is a power of 2.
The algorithm runs recursively until the base case when $m'\le \sqrt{m'+n'}$, where $m'$ and $n'$ are the sizes of the two input trees in the current recursive call, respectively. For the base cases, we apply a weaker (work-inefficient) algorithm discussed in
\ifx\fullversion\undefined the full version of this paper \cite{blelloch2019optimal}.
\else
Appendix~\ref{sec:set:inefficient}.
\fi
The work-inefficient approach will not affect the overall asymptotic bound because of the criterion at which the base cases are reached.
After that, the $d$ pieces are connected using the pivots.
At this time, rebalancing may occur, but we do not handle it immediately.
Instead, we apply a final step at the end of the algorithm to recursively rebalance the output tree based on a reconstruction-based algorithm discussed in Section~\ref{sec:set:rebalance}.
The high-level idea is that, whenever a subtree has two children imbalanced by more than some constant factor (i.e., one subtree is much larger than the other one), the whole subtree gets flattened and reconstructed.
Otherwise, the subtree can be rebalanced using a constant number of rotations.
An illustration of our algorithm is shown in Figure~\ref{fig:setalgorithm}.
\ifx\fullversion\undefined Due to page limitation, we put the algorithm description of base case algorithms, and the cost analysis of the algorithm in the full version of the paper \cite{blelloch2019optimal}, and only briefly show some intuition in Section \ref{sec:set:outline}.
\else
Due to page limitation, we put the algorithm description of base case algorithms (Appendix \ref{sec:set:inefficient}), and the cost analysis of the algorithm (Appendix \ref{sec:set:proof}) in the appendix, and only briefly show some intuition in Section \ref{sec:set:outline}.
\fi

\subsection{Background and Related Work}
\label{sec:set:related}

Ordered set-set operations \union{}, \intersection{} and \difference{} are fundamental algorithmic primitives, and there is a rich literature of efficient algorithms to implement them.
For two ordered sets of size $n$ and $m\le n$, the lower bound on the number of comparisons (and hence work or sequential time) is $\Omega\left(\boundcontent\right)$~\cite{hwang1972simple}.    The lower bound on span in the \bfmodel{} is $\Omega(\log n)$.
Many sequential and parallel algorithms match the work bound~\cite{Blelloch1998,blelloch2016just,akhremtsev2016fast,brownT80}.   In the parallel setting, some algorithms achieve $O(\log n)$ span on the PRAM~\cite{PVW83,PP01}. However, they are not work-efficient, requiring $O(m \log n)$
work. There is also previous work focusing on I/O efficiency \cite{bender2019small} and concurrent operations \cite{fatourou2019persistent,braginsky2012lock} for parallel trees, and parallel data structures supporting batches \cite{gilbert2019parallel,workingset,milman2018bq}.
%Katajainen~\cite{katajainen1994efficient} claimed an algorithm
%with $\bound$ work and $O(\log n)$ span using 2-3 trees, but it
%appears to contain some bugs in the analysis~\cite{Blelloch1998}.

Some previous algorithms achieve optimal work and polylogarithmic span.
Blelloch and Reid-Miller proposed algorithms on treaps with optimal expected work and $O(\log
n)$ span~\whp{} on an EREW PRAM with scan operations, which translates to $O(\log^2 n)$ span in the \bfmodel{}.  Akhremtsev and Sanders~\cite{akhremtsev2016fast} described an algorithm for array-tree \union{} based on $(a,b)$-trees with optimal work and $O(\log n)$
span on a CRCW PRAM.
Blelloch et al.~\cite{Blelloch1998} proposed ordered set algorithms for a variety of balancing schemes \cite{blelloch2016just} with optimal work.
All the above-mentioned algorithms have $O(\log m \log n)$ span in the \bfmodel{}. There have also been parallel bulk operations for self-adjusting data structures~\cite{workingset}.
As far as we know, there is no parallel algorithm for ordered set functions (\union{}, \intersection{} and \difference{}) with optimal work and $O(\log n)$ span in the \bfmodel{}.
%And for the PRAM only randomized algorithms are known that are work efficient (in expectation) and have $O(\log n)$ span~\cite{akhremtsev2016fast}.

\begin{figure*}
    \centering
  \vspace{-.5em}
  \includegraphics[width=0.9\columnwidth]{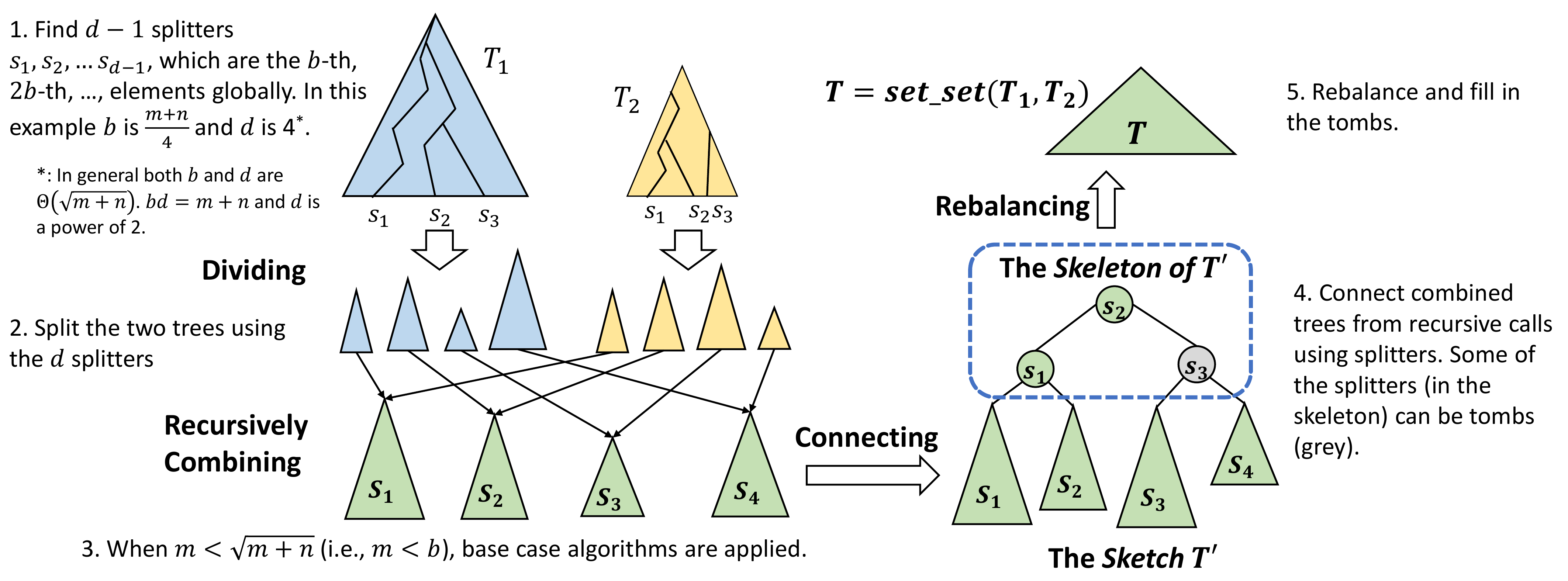}
  \vspace{-1em}
  \caption{\small An illustration of the set-set algorithms. We first split both trees into \tchunk{s} by the glaobally $b$-th, $2b$-th, ..., elements. Here $b=\frac{m+n}{4}$, but in general $b=(n+m)/d$ should be $\Theta(\sqrt{m+n})$ where $d=\Theta(\sqrt{m+n})$ is a power of 2. We then recursively sketch each pair of \tchunk{s}, until we reach the base case and call the base case algorithms. We then connect the results with pivots, and get the \emph{sketch} of the result tree. Finally we rebalance the tree structure and fill in all tombs.}\label{fig:setalgorithm}
  \vspace{-.5em}
\end{figure*}

\subsection{Preliminaries}\label{sec:set:prelim}
Given a totally ordered universe $U$, the problem is to take the union, intersection, and difference of two subsets of $U$.  We assume the comparison model over the \emph{elements} of $U$, and require that the inputs and outputs can be enumerated in-order with no additional comparison (i.e., no cheating by being lazy).

We assume the two inputs have sizes $m$ and $n\ge m$ stored in weight-balanced binary trees~\cite{nievergelt1973binary} with balancing parameter $\alpha$ (WBB[$\alpha$] tree). The weight of a subtree is defined as its size plus one, such that the weight of a tree node is always the sum of the weights of its two children.
WBB[$\alpha$] trees maintain the invariant that for any two subtrees of a node, the weights are within a factor of $\alpha$ ($0 < \alpha \leq 1-1/\sqrt{2}$) of each other.
For the two input trees, we refer to the tree of size $n$ as the \emph{large tree}, denoted as $T_L$, and the tree of size $m$
as the \emph{smaller tree}, denoted as $T_S$. We present two definitions as follows.
Note that these two definitions are more general than the definitions of ancestors and descendants, \emph{since $k$ may or may not appear in $T$}.
%We use a function \texttt{connect}$(T_L,k,T_R)$ to create a tree structure with root element $k$, and left (right) subtree pointers to $T_L$ ($T_R$). We do so regardless of the tree is balanced or not.

\begin{definition}
In a tree $T$, the \emph{\upper{s}} of an element $k \in U$, are all the nodes in $T$ on the search path
to $k$ (inclusive).
\end{definition}

\begin{definition}
In a tree $T$, an element $k \in U$ \emph{falls into} a subtree $T_x\in T$, if the search path to $k$ in $T$ overlaps the subtree $T_x$.
\end{definition}

%As a special case, an \upper{} of an element $k\in T$ is equivalent to an ancestor of $k$, and an element $k \in T$ fall into a subtree $T_x$ iff $k$ is a descendant of $T_x$.

%Lemma \ref{lem:rotation} is a special case of Lemma 3 in \cite{blelloch2016just}.

\paragraph{Persistent Data Structures.} In this section, we use underlying \emph{persistent}~\cite{persistence} (and actually purely functional) tree structure, which uses path-copying to update the weight-balanced trees.  This means that when a change is
made to a node $v$ in the tree, a copy of the path to $v$ is made, leaving the old path and old value of $v$ intact.   Figure \ref{fig:persistence} shows an example of inserting a new element into the tree.  Such a persistent insertion algorithm also copies nodes that are involved in rotations since their child pointers change.

In particular, our algorithm will use a persistent \mf{Split}$(T,k)$ function on WBB$[\alpha]$ trees as discussed in \cite{blelloch2016just,pam}. This function splits tree $T$ by key $k$ into two trees and a bit, such that all keys smaller than $k$ and larger than $k$ will be stored the two output trees, respectively, and the bit indicates if $k\in T$. Because of path-copying, the persistent \mf{Split} returns two output trees and leaves the input tree intact. This algorithm costs $O(\log n)$ work on a tree of size $n$.
%The use of persistence means that none of our algorithms destroy their input.
%    Taking the union of two sets, for example, creates a new set leaving the two original sets intact.

\begin{figure}
\begin{minipage}{.35\columnwidth}
    \centering
  \includegraphics[width=.9\columnwidth]{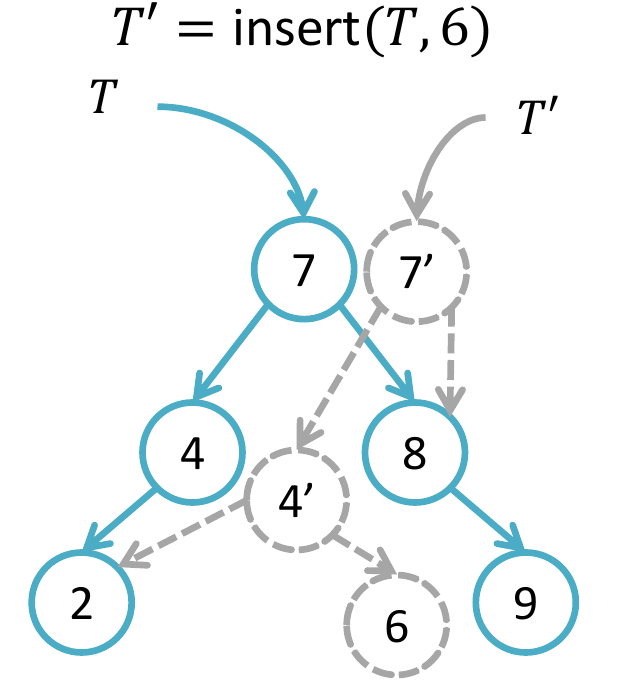}
\end{minipage}
\begin{minipage}{.6\columnwidth}
  \caption{\small A persistent insertion on a tree. The algorithm basically copies all tree nodes on the insertion path, such that the new (copied) root represents the output tree, and the input tree is intact represented by the old root pointer. This algorithm costs $O(\log n)$ time for an input tree of size $n$.}\label{fig:persistence}
\end{minipage} % and copies at most $O(\log n)$ tree nodes
\vspace{-1em}
\end{figure}

% \guy{path copying comes out of the blue.  Persistence should either be discussed more or less.    Do we need it?}
% \yihan{We need it for all the split functions because after multiple split function we need to get several ``isolated'' trees. I added some more information here. }

%This requires creating a node for element $k$ when ever a \texttt{connect}$(T_L,k,T_R)$ is called.

\input{app-code}
\input{setset/algo}
\input{setset/rebalance}

\input{setset/basesummary.tex}

%% file: app-code.tex
\begin{algorithm}[!t]
\caption{$T\gets \mf{Set\_Set}(T_1, T_2)$, the main algorithm for ordered set-set operations}
\label{algo:setset}\small
\SetKwFor{ParForEach}{parallel foreach}{do}{endfch}
\SetKwFor{ParFor}{parallel for}{do}{endpfor}
\SetKw{Break}{break}
  \DontPrintSemicolon
    \vspace{0.5em}
    \KwIn{Two weight-balanced trees storing two ordered sets. }
    \KwOut{A weight-balanced tree $T$ storing the union/intersection/difference of the two input sets.}
    \lIf {$|T_1|<|T_2|$} {\Return {\mf{Set\_Set}$(T_2, T_1)$}}
    $T' \gets \mf{Sketch}(T_1,T_2)$\tcp*[f]{Algorithm~\ref{algo:sketch}}\\
    $T \gets \mf{Rebalance}(T',\mf{false})$\tcp*[f]{Algorithm~\ref{algo:rebalance}}\\
    \Return {$T$}
    \vspace{0.5em}
\end{algorithm}

\begin{algorithm}[!th]
\caption{$T'\gets \mf{Sketch}(T_1,T_2)$}
\label{algo:sketch}\small
\SetKwFor{ParForEach}{parallel foreach}{do}{endfch}
\SetKwFor{ParFor}{parallel for}{do}{endfor}
\SetKw{Break}{break}
  \DontPrintSemicolon
    \vspace{0.5em}
    \KwIn{Two WBB$[\alpha]$ trees. $T_1$ is from the original larger tree $T_L$ and $T_2$ is from the original smaller tree $T_S$.}
    \KwOut{A binary tree sketch $T'$ representing the union/intersection/difference of the two input sets.}
    \vspace{0.5em}
    Let $n'\gets |T_1|$ and $m'\gets |T_2|$\\
    \lIf {$n'$ is $0$} { \Return{$T_2$} }
    \lIf {$m'$ is $0$} { \Return{$T_1$} \label{line:basecaset1}}
    \lIf {$m'<\sqrt{n'+m'}$} { \Return{$\mf{Base\_Case}(T_1, T_2)$}  \label{line:basecase}}
    $d\gets 2^{\lceil\log_2\sqrt{m'+n'}\rceil}$\\
    $b\gets (m'+n')/d$\\
    Let $\mb{splitter}_0 \gets -\infty$ and $\mb{splitter}_{d} \gets +\infty$\\
    \ParFor {$i \gets 1$ to $d-1$} {
        Find $\mb{splitter}_i$, which is the $(i\cdot b)$-th element in $T_1$ and $T_2$ (duplicate value counts twice) by dual-binary search\\
        Let $f_i$ indicate if $\mb{splitter}_i$ is a tomb\\
    }
    \ParFor {$i \gets 1$ to $d$} {
        Split $T_1$ using $\mb{splitter}_{i-1}$ and $\mb{splitter}_i$, output tree $T_{1,i}$\\
        Split $T_2$ using $\mb{splitter}_{i-1}$ and $\mb{splitter}_i$, output tree $T_{2,i}$\label{line:split}\\
        %\tcp*[h]{splitters are not in the output trees}\\
        $T'_{i}\gets \mf{Sketch}(T_{1,i},T_{2,i})$
    }
    Connect $T'_{1},\ldots,T'_{d}$ using $\mb{splitter}_{1},\ldots,\mb{splitter}_{d-1}$\label{line:connect} \\
    \Return{the result tree}\\
\end{algorithm}

\begin{algorithm}[!th]
\caption{$\langle T,e\rangle\gets \mf{Rebalance}(T', \mb{last})$}
\label{algo:rebalance}\small
\SetKwFor{ParForEach}{parallel foreach}{do}{endfch}
\SetKwFor{ParFor}{parallel for}{do}{endpfor}
\SetKw{Break}{break}
\KwIn{A tree sketch $T'$, and a boolean flag \mb{last} indicating if the last element should be extracted.}
\KwOut{A valid weight-balanced tree $T$ with no tombs. If \mb{last} is true, $e$ is the last element extracted from $T'$.}
\SetKwInOut{Note}{Note}
\Note{$\mf{EffectiveSize}(T)$ returns the number of non-tombs in $T$.}
  \DontPrintSemicolon
    \vspace{0.3em}
    %\lIf {$T'=\varnothing$} {\Return {$\langle \varnothing, \varnothing\rangle$}}
    \lIf {$\mf{EffectiveSize}(T')$ is $0$} {\Return {$\langle\varnothing,\varnothing\rangle$}}
    %\lIf {$T'$ is obtained by base cases} {\Return {$T'$}}
    %(\tcp*[f]{check small$(T')$})
    \If {$T'$ is obtained by base cases (Line \ref{line:basecase} in Algorithm \ref{algo:sketch}}
    {\label{line:skipbase}
    \lIf {last} {\Return{$\langle \mf{RemoveLast}(T'), \mf{Last}(T')\rangle$}}\label{line:removelast1}
    \lElse{\Return {$\langle T', \varnothing\rangle$}}
    }
    \If {no nodes in $T_S$ fall into $T'$ (by checking small$(T')$)}
    {\label{line:skip1}
    \lIf {last} {\Return{$\langle \mf{RemoveLast}(T'), \mf{Last}(T')\rangle$}}\label{line:removelast2}
    \lElse{\Return {$\langle T', \varnothing\rangle$}}
    }
    %\lIf {all nodes in $T'$ are tombs} {\Return {$\langle\varnothing,\varnothing\rangle$}}
    \lIf* {\mb{last}} {$b\gets 1$} \lElse {$b\gets 0$}
    \If {$\mf{EffectiveSize}(\mf{LeftTree}(T'))+1$ and $\mf{EffectiveSize}(\mf{RightTree}(T'))-b+1$ differs by more than a factor of $2/\alpha$} {\label{line:ifbalance}
    Flatten $T'$ and reconstruct it (if \mb{last} then extract the last element in $T'$)~\label{line:reconstruct}\\
    \Return{the new tree}\label{line:reconstructreturn}}
    %$rt\gets$ the root of $T'$\\
    \lIf* {the root of $T'$ is a tomb} {$t \gets \texttt{true}$} \lElse {$t \gets \texttt{false}$}
    \SetKwProg{inparallel}{In parallel:}{}{}
    \inparallel{}{
    $\langle T_l, e_l\rangle=\mf{Rebalance}(\mf{LeftTree}(T'), t)$\\
    $\langle T_r, e_r\rangle=\mf{Rebalance}(\mf{RightTree}(T'), \mb{last})$\\
    }
    %\textbf{In parallel:}\\
    %\quad$\langle T_L, e_L\rangle=\mf{Rebalance}(\mf{LeftTree}(T'), t)$\\
    %\quad$\langle T_R, e_R\rangle=\mf{Rebalance}(\mf{RightTree}(T'), \mb{last})$\\
    \If {$\mf{EffectiveSize}(\mf{RightTree}(T'))$ is 0 \textbf{and} \mb{last}} {
    $T\gets T_l$\\
    \lIf {the root of $T'$ is a tomb} {$e \gets e_l$}
    \lElse {$e \gets T'.\mb{root}$}
    }
    \Else{
    $e\gets e_r$\\
    \lIf {the root of $T'$ is a tomb} {$T \gets \mf{Connect}(T_l,e_l,T_r)$}
    \lElse {$T \gets \mf{Connect}(T_l,T'.\mb{root},T_r)$}
    }
    \Return {$\langle\mf{RebalanceByRoatation}(T), e\rangle$}\label{line:rotation}
\end{algorithm}

%% file: setset/algo.tex
\subsection{The Main Algorithms}
\label{sec:set:algo}
We first give a high-level description of our algorithms for the three set-set functions.  As mentioned, we denote the larger input tree as $T_L$, and the smaller input tree as $T_S$. We use two steps, \defn{sketching} and \defn{rebalancing}. The sketching step aims at combining the elements in the two input trees in-order into one tree, which is not necessarily balanced.  The rebalancing step will apply a top-down algorithm to rebalance the whole tree by the WBB[$\alpha$] criteria.

Our sketching algorithm is based on a $d$-way divide-and-conquer scheme, where $d=\Theta(\sqrt{n+m})$ is a power of 2. It is a recursive algorithm, for which the two input trees are denoted as $T_1$ and $T_2$. In particular, $T_1$ contains a subset of $T_L$ and $T_2$ contains a subset of $T_S$. The algorithm will combine the two subsets and return one result tree. Note that even though $T_1\subseteq T_L$ and $T_2 \subseteq T_S$, the sizes of $T_1$ is not necessarily larger than $T_2$.
Throughout the recursive process, we track the following quantities for each tree node $v$:
\begin{enumerate}
  \item The size of the subtree, noted as $\text{size}(v)$.
  \item The number of elements originally from $T_L$, noted as $\text{large}(v)$.
  \item The number of elements originally from $T_S$, noted as $\text{small}(v)$.
  \item The number of elements appearing both in $T_L$ and $T_S$, noted as $\text{common}(v)$.
\end{enumerate}
The tree size $\text{size}(v)$ is required by the WBB[$\alpha$] invariant. The other three are used for $T_L-T_S$, $T_S-T_L$, and $T_S \cap T_L$ respectively.
The generic algorithm for all three operations is given in Algorithms~\ref{algo:setset},~\ref{algo:sketch},  and~\ref{algo:rebalance}. An illustration is shown in Figure~\ref{fig:setalgorithm}.     The difference between the three set-set functions is only in the base cases.
%It is generic for all the three set operations, and the difference only applies in the base cases.

We now present the two steps of the algorithm in details:
%algorithm in more details. As mentioned, the recursive algorithm has two steps:
\begin{enumerate}[leftmargin=*]
  \item \textbf{Sketching} (Algorithm \ref{algo:sketch}).
        This step generates an output tree $T'$ with all elements in the result, although not rebalanced.
        There are three subcomponents in this step. Denote $n'=|T_1|$ and $m'=|T_2|$, which means the number of tree nodes handled by this recursive call that are originally from the larger and smaller tree, respectively. As mentioned, $m'$ can be even larger than $n'$ in some of the recursive calls.
      \begin{enumerate}[labelindent=0pt]
      \item \textbf{Base Case.} When $m'<\sqrt{n'+m'}$, the algorithm reaches the base case.
          It calls the work-inefficient algorithm to generate a balanced output tree in $O(m'\log n')$ work and $O(\log n')$ span, which are presented in
          \ifx\fullversion\undefined the full version of this paper \cite{blelloch2019optimal}.
          \else
          Appendix \ref{sec:set:inefficient}.
          \fi
      \item \textbf{Dividing.} We then use $d-1$ \defn{pivots} to split both input trees into $d$ \defn{\tchunk{s}}, and denote the partitioning of $T_1$ as $T_{1,\{1,\ldots,d\}}$, and $T_2$ as $T_{2,\{1,\ldots,d\}}$.
          The $d-1$ pivots are the global $b$-th, $2b$-th, \dots elements in the two trees, where $b=(n+m)/d$, so that $|T_{1,i}|+|T_{2,i}|$ for all $i$ have the same value (or differ by at most 1).
          All the splits (Line \ref{line:split}) can be done in parallel using a persistent \texttt{split} algorithm on weight-balanced trees~\cite{blelloch2016just}.  We then apply the
algorithm recursively on each pair of \tchunk{s}.

          Note that not all the pivots should appear in the output tree of the entire algorithm, depending on the set function.
          For example, for \intersection{}, those pivots that only appear in one tree will not show up at the end.
          In this case, in the \mf{Sketch} step, we will mark such pivot nodes as \emph{tomb}s, and filter them out later in the rebalancing step.

          %Note that if a pivot is marked as a tomb, the \tchunk{} right before it is responsible for extracting its last element to replace the tomb. Here we also need to deal with more details if the corresponding \tchunk{} is empty. We will address that later.
      \item \textbf{Connecting.} After the dividing substep and recursive calls, we have $d-1$ pivots (including tombs), and $d$ combined \tchunk{s} returned by the recursive calls.
          In the connecting substep, we directly connect them regardless of balance. Since $d$ is a power of 2, the $d-1$ pivots will form a full balanced binary tree structure on the top $\log_2 d$ levels, and all the \tchunk{s} output from recursive calls will dangle on the $d$ pivots. This process is shown in Line~\ref{line:connect}.
      \end{enumerate}
      The output $T'$ of the \mf{Sketch} step is a binary tree, which may or may not be balanced.
      We will call $T'$ the \defn{sketch} of the final output of the algorithm.  We also call the top $\log_2 d$ levels in $T'$ consists of pivots the \defn{skeleton} of $T'$.
      We note that $T'$ may contain tombs, and we will filter them out in the next step.
  \item \textbf{Rebalancing}. We will use a reconstruction-based rebalancing algorithm to remove the tombs and rebalance the sketch tree $T'$ (Algorithm~\ref{algo:rebalance}, see more details in Section~\ref{sec:set:rebalance}). This rebalancing algorithm is stand-alone, and is of independent interest. %We describe the algorithm in Section~\ref{sec:set:rebalance}.
        %This step rebalances the tree obtained by the sketching step. We recursively \emph{settle} each level by rebalancing them. For a tree node, if its two subtrees are not balanced, we flatten the tree structure and re-build the whole tree. Otherwise, we recursively settle its two children, and after that directly connect the two children. This is very similar to the rebalancing step of the work-inefficient algorithms.
\end{enumerate}

%% file: setset/rebalance.tex
\subsection{The Rebalancing Algorithm}
\label{sec:set:rebalance}
We now present the reconstruction-based rebalancing algorithm.  A similar idea was also used in~\cite{BGSS18}. In this paper, we use this technique to support better parallelism instead of write-efficiency.

We use the \emph{effective size} of a subtree as the number of elements in this subtree excluding all tombs.
The effective size for a tree node $v$ can be computed based on size($v$), large($v$), small($v$) and common($v$), depending on the specific set operation.
It is used to determine if two subtrees will be balanced after removing all tombs.

The rebalancing algorithm is given in Algorithm~\ref{algo:rebalance}.
The algorithm recursively \emph{settles} each level top-down.
For a tree node, we check the effective sizes of its two children and decide if they are almost-balanced.
Here \emph{almost-balanced} indicates that sizes of the two subtrees differ by at most a factor of $2/\alpha$.\footnote{Generally speaking, the constant 2 here can be any value, but here we use 2 for convenience.}
If not, we flatten the subtree and re-build it.
Otherwise, we recursively settle its two children, and after that we re-connect the two subtrees back and rebalance using at most a constant number of rotations.

We also need to filter out \emph{tombs}, since they should not appear in the output tree.
We do this recursively.
If the current subtree root of $T'$ in Algorithm \ref{algo:rebalance} is a tomb, we will need to fill it in using the last element in its left subtree. We note that the effective size of the left subtree cannot be 0 (otherwise the algorithm returns at Line \ref{line:reconstructreturn}). To do this, the algorithm will take an extra boolean argument \texttt{last} denoting if the last element of the result needs to be extracted (returned as $e$ in the output of Algorithm \ref{algo:rebalance}). In this case, if the root of $T'$ is a tomb, the algorithm simply passes a \texttt{true} value to the left recursive call, getting the last element to replace the tomb.

For computing the last value (denoted as $r$), there are two cases.
First, if the subtree needs rebalancing, then after flattening the elements into an array, we simply take out the last element in the array as $r$ and return. Extracting the last element is inlined in the process of reconstruction (Line \ref{line:reconstruct}). Otherwise, we recursively deal with the two subtrees. If \texttt{last} is true, we also extract the last element in its right subtree.

Multiple base cases apply to this rebalancing algorithm. If the effective size of $T'$ is 0, the algorithm directly returns an empty tree and an empty element. The second case is when no element in $T_S$ falls into $T'$. This can be determined by looking at $\mb{small}(T')$. Note that all the \tchunk{s} in the sketching algorithm is designed to be the same size. Therefore, in this case, the whole subtree should be (almost) perfectly balanced, so we directly return it. These base cases are essential in bounding the work of rebalancing, since we do not need to traverse the whole subtree for these special cases.

%% file: setset/basesummary.tex
\subsection{Base Case Algorithms and Cost Proof Outline}
\label{sec:set:outline}
Due to page limitation, we put the algorithm description of base case algorithms and the cost analysis of the algorithm in
\ifx\fullversion\undefined the full version of this paper \cite{blelloch2019optimal}.
\else
the appendix.
\fi
Here we show a very brief description about the intuition.

\textbf{Base case algorithms.} The base case algorithms $\mf{Base\_Case}(T_1, T_2)$ in Algorithm \ref{algo:sketch} use an work-inefficient version ($O(m'\log n')$ work and $O(\log(n'+m'))$ span) of the set-set algorithms. These algorithms are applied when $m'<\sqrt{n'+m'}$, which guarantees the total base case cost is $\bound$. The intuition of the base case algorithms is to search all $m'$ elements from $T_1$ in $T_2$, and based on the set operation being performed, add (remove) the $m'$ elements into (from) $T_2$. The same rebalancing algorithm as in Section \ref{sec:set:rebalance} is applied to guarantee a balanced output tree. Detailed description is in
\ifx\fullversion\undefined the full version of this paper \cite{blelloch2019optimal}.
\else
Appendix \ref{sec:set:inefficient}.
\fi

\textbf{Span.} We will show that all base cases, \mf{Sketch}, and \mf{Rebalance} algorithms have span $O(\log (n+m))$. We first prove that the height of the sketch $T'$ is $O(\log (n+m))$
\ifx\fullversion\undefined in the full version of this paper \cite{blelloch2019optimal}.
\else
(Lemma \ref{lem:step1height}).
\fi
The span of the base cases is straight-forward.
For \mf{Sketch}, this bound holds because of the $\sqrt{m+n}$-way divide-and-conquer. For \mf{Rebalance}, the span holds because the algorithm settles each node top-down, and settling each level in the skeleton only costs a constant span. For the skeleton of the returned tree, if a node is nearly-balanced, then a constant number of rotations settles it. Otherwise, flattening and reconstructing a tree of height $h$ takes $O(h)$ span, which is also equivalent to a constant per level.
In all, the span is $O(\log (n+m))$. We formally prove the span of the algorithm in
\ifx\fullversion\undefined the full paper~\cite{blelloch2019optimal}.
\else
Lemma \ref{lem:setdepth}.
\fi

\textbf{Work bound.} For work, we will prove that all base cases, \mf{Sketch}, and \mf{Rebalance} algorithms cost work $\bound$. We first show that the number of pivots is
\ifx\fullversion\undefined $O(m)$.
\else
$O(m)$ (Lemma \ref{lem:mpivots}).
\fi
%The recurrence of \mf{Sketch} algorithm is leaf-dominated, bounded by all base cases, which is $\bound$.
%The work of the \mf{Sketch} algorithm can be proved inductively.
Most interestingly, for \mf{Rebalance}, the optimality in work lies in the reconstruction-based algorithm.
For all pivots in the skeleton, if it is nearly balanced, the rebalancing cost is a constant. Therefore the total work is proportional to the size of the skeleton, which is no more than $\bound$.

To show the total reconstruction work, in the sketch $T'$, we mark all upper nodes of the elements in $T_S$ as red. There are at most $\bound$ red nodes in
\ifx\fullversion\undefined
$T'$.
\else
$T'$ (Lemma \ref{lem:uppers}).
\fi
We will show that the reconstruction work averaged to each red node is a constant.
%This happens only when at least $c|T_x|$ elements in $T_S$ \emph{falls into} $T_x$ for some constant $c$.
The key observation is that, rebalancing for a subtree $T_x\in T'$ happens only when there are $m_x\ge c|T_x|$ red nodes in $T_x$, where $c$ is a constant.
%only when the number of elements from $T_S$ \emph{falling into} $T_x$, call it $m_x$, is at least $c|T_x|$, for some constant $c$.
This is because the two subtrees of $T_x$ are supposed to have the same size ($\sqrt{n'+m'}$) due to the selection of pivots. However there can be duplicates in \union{}; also \intersection{} and \difference{} do not keep all input elements in the output. Therefore there can be imbalanced in size.
%This is because when we look at the elements $T_x$ from $T_L$, they are balanced originally.
We will show that the size of either subtree changes by no more than $m_x$. Therefore, to make them unbalanced, $m_x$ has to be at least $c|T_x|$ for some constant $c$. This makes the average cost per red node to be $O(1)$.
Adding the cost of all red nodes gives the stated optimal work bound.
We will formally prove the work of the algorithm in
\ifx\fullversion\undefined the full version of this paper~\cite{blelloch2019optimal}.
\else
Lemma \ref{lem:setwork}.
\fi

%% file: rand.tex
\section{Random Permutation}\label{sec:permutation}

Generating random permutation in parallel is useful in parallel algorithms, and is used in the list and tree contraction algorithms in this paper.
Hence it has been well-studied both theoretically~\cite{Alonso1996,Anderson1990,Czumaj96,Gibbons1996,Gil91,Gil91a,Gustedt03,Hagerup91,Miller85,RR89,shun2015sequential}
and experimentally~\cite{CongB05,Gustedt2008,shun2015sequential}.
To the best of our knowledge, none of these algorithms can be implemented in the \bfmodel{} using linear work and $O(\log n)$ \depth{}.
We now consider the simple sequential algorithm of Knuth~\cite{Knuth69} (Durstenfeld's~\cite{Durstenfeld64}) shuffle that iteratively decides each element:

{
%\vspace{-1em}
\setlength{\interspacetitleruled}{0pt}%
\setlength{\algotitleheightrule}{0pt}%
\begin{algorithm}[h]
%\tcc{}
\Fn {\mf{KnuthShuffle}(A, H)} {
$A[i]\gets i$ for all $i=0,\ldots,n-1$\\
\For {$i \leftarrow n-1$ to $0$} {
    swap($A[i]$, $A[H[i]]$)
}
}
\end{algorithm}
%\vspace{-1em}
}

\noindent where $H[i]$ is an integer uniformly drawn between 0 and $i-1$, and $A[\cdot]$ is the output random permutation.

A recent paper~\cite{shun2015sequential} shows that this sequential iterative algorithm is readily parallel.
The key idea is to apply multiple swaps in parallel as long as the sets of source and destination locations of the swaps are disjoint.
Figure~\ref{fig:rand} shows an example, and we can swap location 5 and 2, 7 and 1, 6 and 3 simultaneously in the first round, and the three swaps do not interfere each other.
If the nodes pointing to the same node are chained together and the self-loops are removed, we get the dependences of the computation.  An example is given in Figure~\ref{fig:rand}(b).
%If we chain the roots of the dependence trees, then we get a binary tree.
Similar to list contraction, we can execute the swaps for all leaf nodes and remove them from the tree in a round-based manner.
It can be shown that the modified dependences by chaining all the roots in the dependence forest (as shown in Figure~\ref{fig:rand}(c)) correspond to a random binary search tree, and the tree depth
is again bounded by $O(\log n)$~\whp{}.
The \depth{} of this algorithm is therefore $O(\log^2 n)$~\whp{} in the \bfmodel{}.

Similar to the new list contraction algorithm discussed in
Section~\ref{sec:list}, the computation can be executed
asynchronously.  Namely, the swaps in different leaves or subtrees are
independent.  Therefore, once the dependence structure is generated,
we can apply a similar approach as in Algorithm~\ref{algo:lc}, but
instead of splicing out each node, we swap the values for the pair of
nodes.

The remaining question is how to generate the dependence structure.
We do this in two steps.  We first semisort all nodes based on the
destination locations (grouping the nodes on all the horizontal chains
in Figure~\ref{fig:rand}(b) or right chains in
Figure~\ref{fig:rand}(c)).  Then we use an algorithm that takes
quadratic work to
sort the nodes within each group, and connect the nodes as discussed.

\input{figure-rand}

\begin{theorem}
  The above algorithm generates a random permutation of size $n$ using $O(n)$ expected work and
  $O(\log n)$ \spanc{} \whp{} in the \bfmodel{}.
\end{theorem}

\begin{proof}
Similar to the list contraction algorithm in Section~\ref{sec:list}, this algorithm applies the same operations as the random permutation algorithm in~\cite{shun2015sequential}, and the swaps obey the same ordering for any pair of nodes with dependency.
The improvement for \depth{} is due to allowing asynchrony for disjoint subtrees.

The cost after the construction of dependence tree is the same as the list contraction algorithm (Algorithm~\ref{algo:lc}), which is $O(n)$ work and $O(\log n)$ \spanc{} \whp{}.
For constructing the dependence tree, the semisort step takes $O(n)$ expected work and $O(\log n)$ \depth{} \whp{} using the algorithm in Section~\ref{sec:sorting}.
The quadratic work sorting can easily be implemented in $O(\log n)$ \depth{} \whp{},
as in Section~\ref{sec:sorting}.
We now analyze the work to sort the chains.

Let a 0/1 random variable $A_{i,j}$ is 1 if $H[i]=j$ for $j<i$, and the probability $\Pr[A_{i,j}=1]$ is $1/i$.
$\Pr[A_{i,j}A_{k,j}=1]$ is then $1/(ik)$ for $j<i<k$ since they are independent.  %\guy{since they are independent.}
The expected overall work for sorting is (omitting constant in front of $A_{i,j}$).
%\guy{Missing constants and inequalities.   Also not sure how it is derived.}
%\yihan{adding more steps}
\begin{align*}
  &\E[W_{\smb{RandPerm}}(n)] = \E\left[\sum_{j=1}^{n}\left(\sum_{i=j}^{n}{A_{i,j}}\right)^2\right] \\
  =&\E\left[\sum_{j=1}^{n}\sum_{i=j}^{n}{A_{i,j}^2}\right]+2\cdot\E\left[\sum_{j=1}^{n}\sum_{i=j+1}^{n}\sum_{k=i+1}^{n}{A_{i,j}A_{k,j}}\right] \\
  =&  \sum_{i=1}^{n}\sum_{j=1}^{i-1}{1/i}+\sum_{k=1}^{n}\sum_{i=1}^{k-1}\sum_{j=1}^{i-1}{1/ik} ~=~O(n)\\
  %&\le& \!\! \sum_{i=1}^{n}{(i-1)/i}+\sum_{i=1}^{n}\sum_{k=i+1}^{n}{(i-1)/ik}\\
  %O\left(n\cdot 1+\sum_{i=1}^{n}\sum_{k=i+1}^{n}1/k\right)
\end{align*}
Combining all results gives the stated theorem.
\end{proof}

%% file: figure-rand.tex
\begin{figure}[t]
\begin{center}
  \includegraphics[width=0.9\columnwidth]{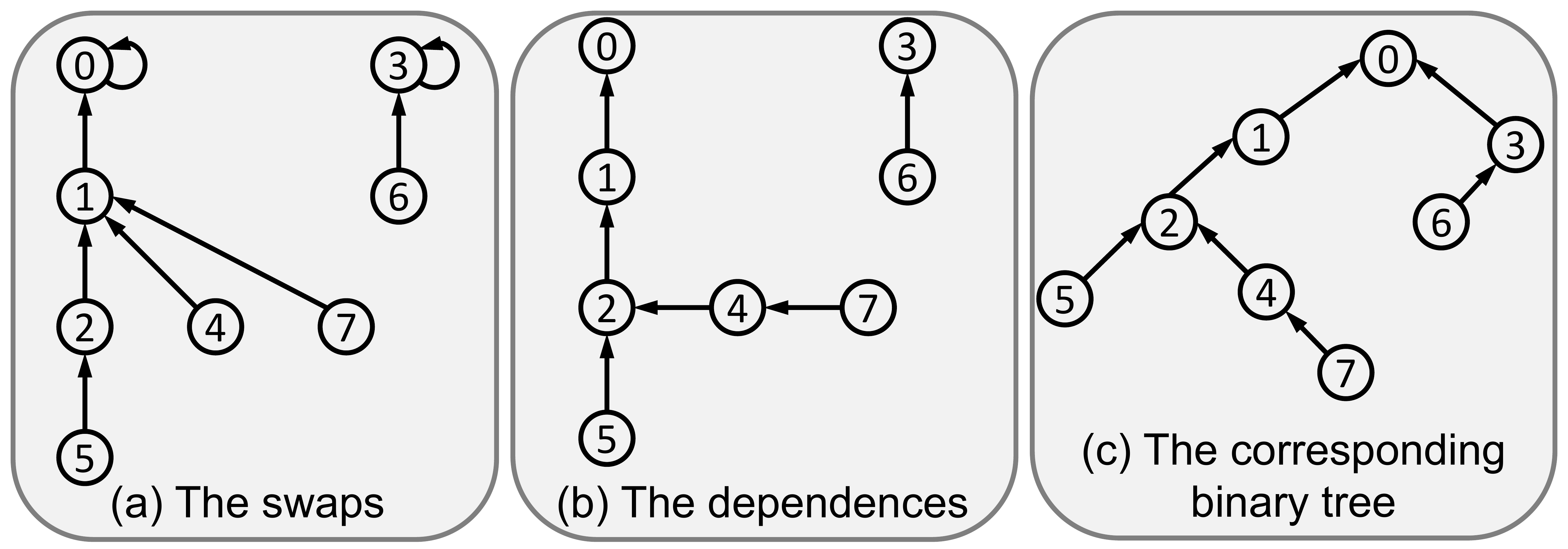}
\end{center}\vspace{-1.2em}
\caption{\small An example when $H = [0, 0, 1, 3, 1, 2, 3, 1]$.  Figure~(a)
  indicates the destinations of the swaps shown by $H$.  The
  dependences of the swaps are shown by Figure~(b), indicating the
  order of the swaps.  Figure~(c) links the roots of the forest to
  make it a binary tree.}
\vspace{-2em}
\label{fig:rand}
\end{figure}

%% file: figure-tree.tex
%\begin{figure*}[t]
%\begin{center}
%  \includegraphics[width=1.8\columnwidth]{figures/tree}
%\end{center}\vspace{-1.2em}
%\caption{An example of the tree contraction algorithm.  It first generates the priority of the leaf node from a random permutation, as shown in Figure~(a).
%Then every interior node computes the highest priority label in the subtree, and the computed values are shown as the blue numbers in Figure~(b).
%Then each interior node $v$ is paired with the leave node that has the priority which is stored in the child node different from the priority of $v$.
%Taking the root node as an example, the priority 5 is different from the right node which has priority 4, so the root node is paired with the leaf node with priority 4.
%The pairing results are shown as the red numbers in Figure~(c).
%In this example, leaf nodes 0, 1, and 2 can be contracted simultaneously, and the contracted tree is shown in Figure~(d).
%}
%\label{fig:tree}
%\end{figure*}

\begin{figure*}[!h!t]
\begin{center}
  \includegraphics[width=\columnwidth]{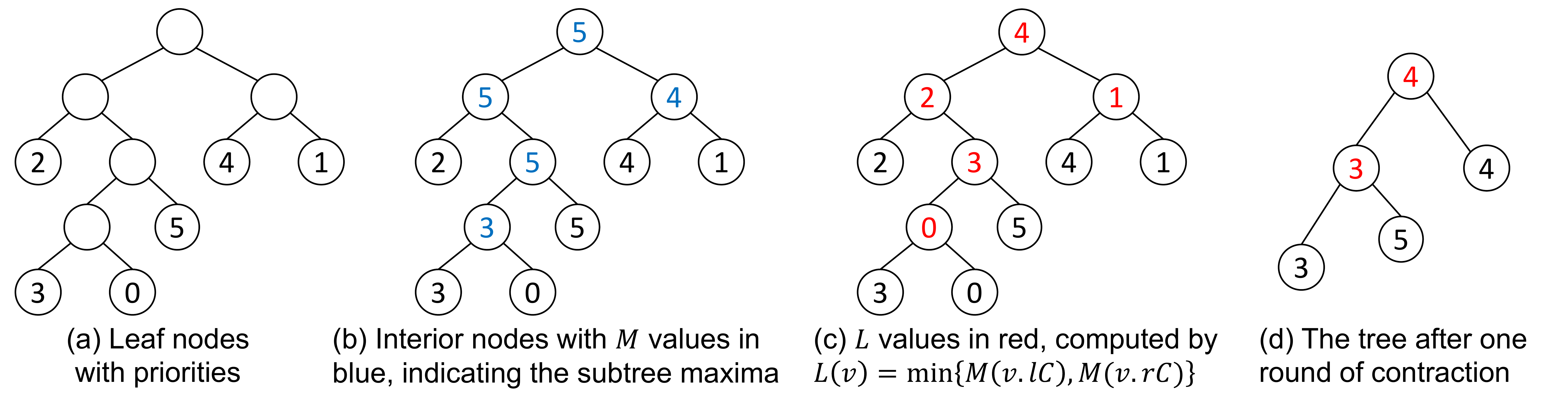}
\end{center}\vspace{-1.2em}
\caption{\label{fig:tree} \small An example of the tree contraction algorithm.  It first generates the priority of the leaf node from a random permutation, as shown in Figure~(a).
Then every interior node computes the highest priority label in the subtree, and the computed values are shown as the blue numbers in Figure~(b).
Then each interior node $v$ is paired with the leaf node that has the priority which is stored in the child node different from the priority of $v$.
Taking the root node as an example, the priority 5 is different from the right node which has priority 4, so the root node is paired with the leaf node with priority 4.
The pairing results are shown as the red numbers in Figure~(c).
In this example, leaf nodes 0, 1, and 2 can be contracted simultaneously, and the contracted tree is shown in Figure~(d).
}
\end{figure*} 

%% file: rmq.tex
\section{Range Minimum Queries}\label{sec:rmq}

Given an array $A$ of size $n$, the range minimum query (RMQ) takes two input indices $i$ and $j$, and reports the minimal value within this range.
RMQ is a fundamental algorithmic building block that is used to solve other problems such as the lowest common ancestor (LCA) problem on rooted trees, the longest common prefix (LCP) problem, and lots of other problems on trees, strings and graphs.
% \guy{Mention previous PRAM algorithms.}
% \yan{better?}
% \guy{The following paragraph is rough and confusing.}

An optimal RMQ algorithm requires linear preprocessing work and space, and constant query time.
It can be achieved by a variety of algorithms (e.g.,~\cite{berkman1993recursive,bender2000lca,alstrup2004nearest,fischer2006theoretical,arge2013dynamic}).
These algorithms are based on the data structure referred to as the sparse table that can be precomputed in $O(n \log n)$ work where $n$ is the input size, and constant RMQ cost.
To further improve the work, the high-level idea in these algorithms is to chunk the array into $O(n/\log n)$ groups each with size $O(\log n)$, find the minima of the groups, and only preprocess the sparse table for the minima.
Within each group, these algorithms use different techniques to preprocess in $O(\log n)$ work per group, and support constant query cost within each group.
Then for a range minimum query $(i,j)$, the minimum can be answered by combining by the query for the sparse table for the whole groups in this range, and the query for the boundary groups that contain elements indexed at $i$ and $j$.
These algorithms can be trivially parallelized in the PRAM model using $O(\log n)$ span (time), but when translating to the \bfmodel{}, the span becomes $O(\log^2 n)$ in preprocessing the sparse table, and needs to be improved.
%To parallelize the preprocessing step in the \bfmodel{}, we only need to consider the step to compute the sparse table, because each group only has $O(\log n)$ size so even sequentially processing within each group only requires $O(\log n)$ \depth{}.

For simplicity, we first assume the number of groups $n'$ is a power of 2.
In the classic sparse table, we denote $T_{i,k}$ as the minimal value between group range $i$ and $i+2^k-1$. It can be computed as $\min\left\{T_{i,k-1},T_{i+2^{k-1},k-1}\right\}$.
Let $k=\lfloor \log_2{(j-i)}\rfloor$. 
Then for query from group $i$ to $j (>i)$, we have $\mb{RMQ}(i,j)=\min\left\{T_{i,k},T_{j-2^k+1,k}\right\}$.
Directly parallelizing the construction for the sparse table uses $O(\log^2 n)$ \depth{}---$O(\log n)$ levels in total and $O(\log n)$ \depth{} within each level.
We now consider a variant of the sparse table which is easier to be generated in the \bfmodel{} and equivalently effective.

In the modified version, we similarly have $\log_2 n'$ levels, and in $k$-th level we partition the array into $n'/2^k$ subarrays each with size $2^k$.
For each subarray, we further partition it to two parts with equal size, and compute the suffix minima for the left side, and prefix minima for the right side.
We denote $T'_{i,k}$ as such value with index $i$ in the $k$-th level.
For each query $(i,j)$, we find the highest significant bit that is different for $i$ and $j\ge i$.
If this bit is the $k$-th bit from the right, then we have $\mb{RMQ}(i,j)=\min\left\{T'_{i,k},T'_{j,k}\right\}$.
An illustration is shown in Figure~\ref{fig:rmq}.

\input{figure-rmq}

We now describe how to compute $T'_{i,k}$.
We note that the computation for each subarray is independent, and each takes linear work and logarithmic \depth{} proportional to the subarray size~\cite{blelloch1990pre}.
Since each element corresponds to $\log_2n'$ computed values, the overall cost is therefore $O(n'\cdot \log_2n')=O(n)$ work and $O(\log n)$ \depth{}.

\begin{theorem}
  The range minimum queries for an array of size $n$ can be preprocessed in $O(n)$ work and $O(\log n)$ \depth{} in the \bfmodel{}, and each query requires constant cost.
\end{theorem}

%% file: figure-rmq.tex
\begin{figure}[t]
\begin{center}
  \includegraphics[width=.7\columnwidth]{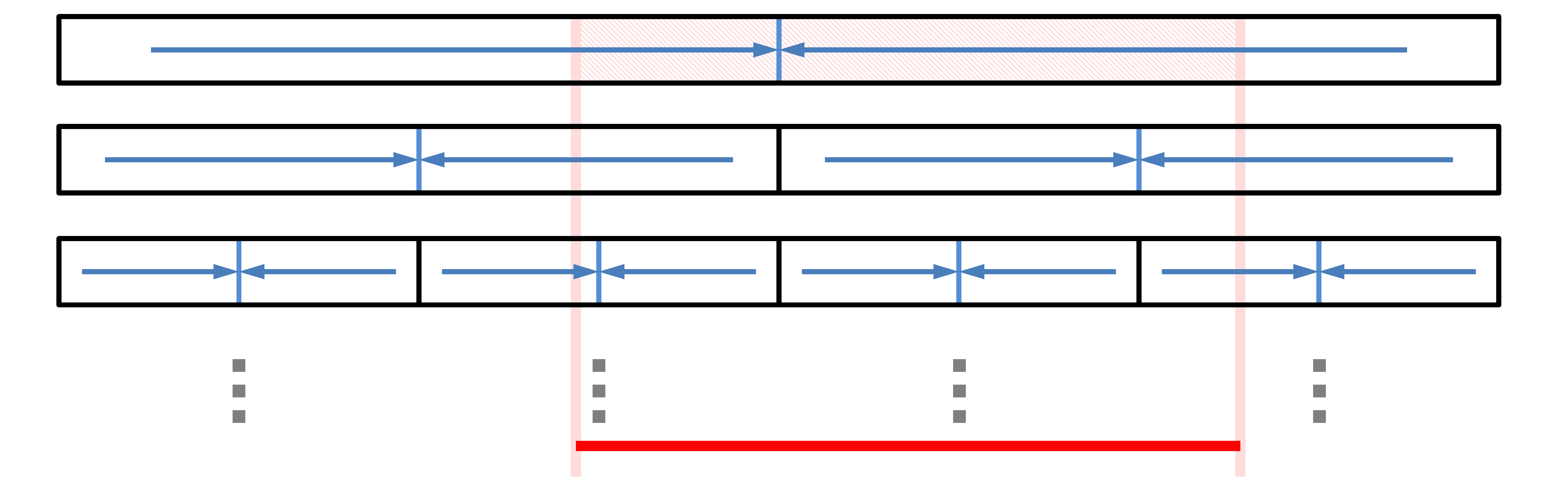}
\end{center}\vspace{-1em}
\caption{\small An illustration of the modified sparse table in Section~\ref{sec:rmq}.
The range is hierarchically partitioned into logarithmic number of levels, and the prefix and suffix minima are computed as the arrows indicate.
For each query range shown as the red segment, we can locate a unique level such that the minimum of the range can be answered by the suffix and the prefix minima (the shaded range).
}
\label{fig:rmq}
%\vspace{-1em}
\end{figure}

%% file: treecontraction.tex
\section{Tree Contraction}\label{sec:tree}

Parallel algorithms for tree contraction have received considerable
interest because of its ample applications for many tree and graph
applications~\cite{JaJa92,Miller85,Reif93,shun2015sequential,BBFGGMS16}.
There are many variants
of parallel tree contraction.  Here we will assume we are contracting
rooted binary trees in which every internal node has exactly two
children.  Any rooted tree can be reformatted to this shape in linear
work and logarithmic span.  We assume the tree $T$ has $n$ leaf nodes
(and $n-1$ interior nodes). We use $v.lC$ and $v.rC$ to denote the left
and the right child of a node $v$, respectively.

List contraction can be considered as a degenerated case of tree
contraction when all interior nodes are chained up.  As a result, we
do not know an optimal parallel algorithm for tree contraction with
$O(n)$ work and $O(\log n)$ span.  Similarly, the difficulty in
designing such an algorithm remains in using no synchronization.

Here we consider parallelizing the sequential iterative algorithm that
``rakes'' one leaf node at a time.  A rake operation removes a leaf
node and its parent node $v$, and if $v$ is not the root, it sets the
other child of $v$ to replace $v$ as the child of $v$'s parent.  We
assign each leaf node a priority drawn from a random permutation, so
the priority defines a global ordering of the nodes to be removed, and
eventually only one node with the lowest priority remains.  By
maintaining some additional information on the tree nodes, we can
apply a variety of tree operations such as expression evaluation,
roofix or leafix, which are useful in many applications~\cite{Reif93}.

Similar to list contraction, we want to avoid applying two rake operations simultaneously such that one of the interior nodes is the parent of the other.
Beyond that, we can rake a set of leaf nodes in parallel.
For instance, in Figure~\ref{fig:tree}(a), we can contract leaf nodes 0, 1, 2, and their parents together, as shown in Figure~\ref{fig:tree}(d).

To decide the nodes that can be processed together, we define $M(v)$ of each interior node $v$ as the lowest priority (maximum value) of any of the leaves in its subtree (blue numbers in Figure~\ref{fig:tree}(b)).
Based on $M(\cdot)$, we further define $L(v)=\min\{M(v.\mb{lC},v.\mb{rC})\}$ (red numbers in Figure~\ref{fig:tree}(c)) if $v$ is an interior node, or its own priority if $v$ is a leaf node.
$L(v)$ defines a one-to-one mapping between the interior nodes and the leaf nodes (except for one leaf node that stays at the end), and $L(v)=u$ indicates that $v$ will be raked by the leaf node $u$.
Based on the labeling, the parallel algorithm in~\cite{shun2015sequential} checks every node $v$, and it can be raked immediately if $v$' parent has an $L$ value smaller than those of $v$'s sibling and $v$'s grandparent (if it exists).
Otherwise the node waits for the next round.
If we rake all possible leaf nodes in a round-based manner, the number of rounds is $O(\log n)$ \whp{}, leading to an $O(\log^2 n)$ span \whp{} in the \bfmodel{}.

Assuming $L(\cdot)$ has already been computed, we can change the round-based algorithm to an asynchronous divide-and-conquer algorithm similar to the list contraction algorithm (Algorithm~\ref{algo:lc}) in Section~\ref{sec:list}. 
The only difference is when setting the $\mb{flag}$s, since now there can be 1, 2, or 3 directions that may activate a postponed node (in list contraction it is either 1 or 2, depending on the initialization of the $\mb{flag}$ array).
This however, can be easily decided by checking the number of neighbor interior nodes.
Similarly, the last thread corresponding to the contraction of a neighbor node that reaches a postponed node activates it and apply the rake operation.
Since the longest possible path has length $O(\log n)$, the algorithm for the contraction phase uses $O(n)$ work, and $O(\log n)$ span \whp{}.

The last challenge is computing $L(\cdot)$.
As shown in Figure~\ref{fig:tree}(b), computing $M(\cdot)$ is a leafix operation on the tree (analogy to prefix minima but from the leaves to the root), which can be solved by the standard range minimum queries as discussed in Section~\ref{sec:rmq}, based on Euler-tour of the input tree.
In Section~\ref{sec:list}, we discussed the list ranking algorithm to generate the Euler tour.
As a result, computing $M(\cdot)$ and $L(\cdot)$ uses $O(n)$ work and $O(\log n)$ span \whp{}.
In summary, we have the following theorem.

\begin{theorem}
  Tree contraction uses $O(n)$ work and $O(\log n)$ span \whp{} in the \bfmodel{}.
\end{theorem}

%% file: setset/base.tex
\section{Base Case Algorithms for Set-set Operations}
\label{sec:set:inefficient}
\begin{figure*}
  \centering
  \includegraphics[width=\columnwidth]{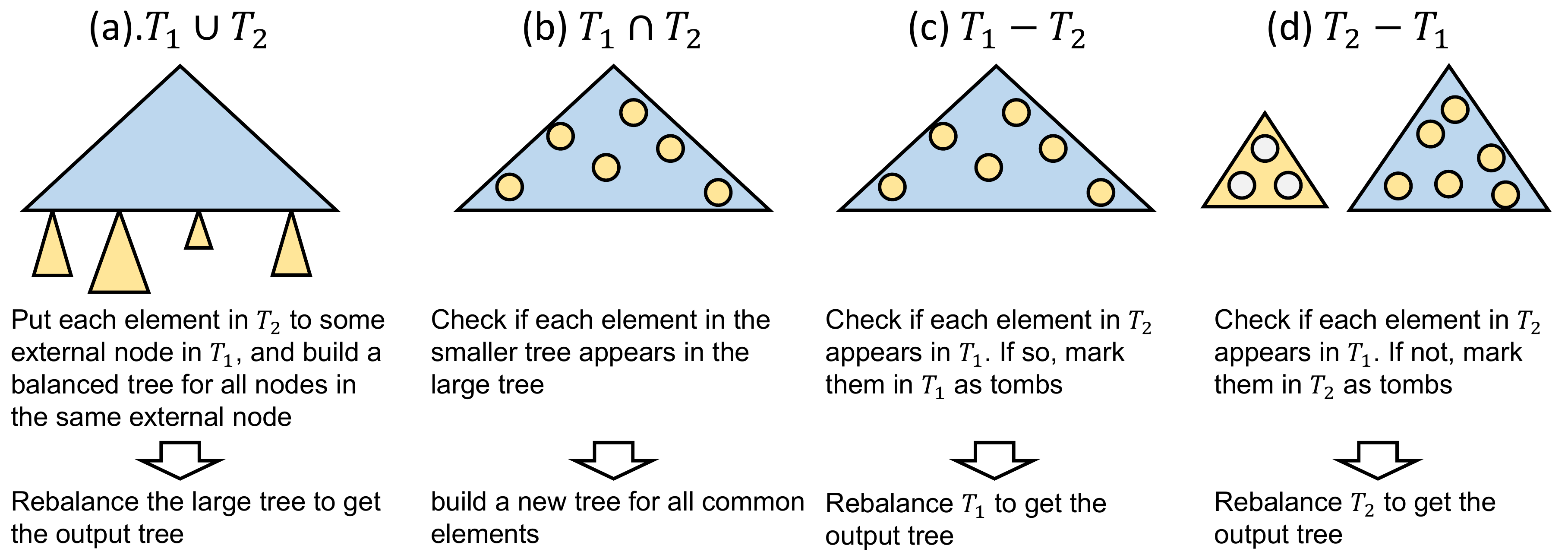}
  \caption{Illustrations for the base case algorithms for \union{}, \intersection{} and \difference{}. Objects in $T_L$ are marked in blue, and objects in $T_S$ are marked in yellow.}\label{fig:basecase}
  %\vspace{-.25in}
\end{figure*}

We now present algorithms for base cases on two input trees $T_1$ and $T_2$. These algorithms solve the set functions in $O(\log n')$ span and $O(m'\log n')$ work, for $|T_1|=n'$ and $|T_2|=m'$. In our algorithms, $T_1$ is from the original larger tree, and $T_2$ is from the original smaller tree. As mentioned, the algorithm is called only when $m'<\sqrt{n'+m'}$.
We give a brief illustration of the algorithms in Figure~\ref{fig:basecase}.
%Recall the subtree from the small tree is $T_S$ with size $m$, and the other tree is $T_L$ with size $n$.

\para{Intersection.} As shown in Figure~\ref{fig:basecase}(b), we will only need to search in $T_1$ for each element in $T_2$, and for those appearing in both sets, we build a new tree structure.
This requires $O(m'\log n')$ work and $O(\log n')$ span.

\para{Union.} As shown in Figure \ref{fig:basecase}(a), our approach consists of two parts: a \emph{scatter} phase, which locates all nodes in $T_1$ in one of the external nodes in $T_2$, and a \emph{rebalancing} phase, which rebalances the tree via our reconstruction-based algorithm (of course we skip line~\ref{line:skipbase} in Algorithm~\ref{algo:rebalance}). Later in Corollary \ref{cor:rebalancework} we show that the work of rebalancing is $O\left(m'\log\left(\frac{n'}{m'}+1\right)\right)$, and the span of rebalancing is $O(\log n')$.

We now focus on the scatter algorithm.
The scatter phase first flattens $T_2$ into an array, and applies $m'$ searches for each element from $T_2$ in $T_1$.
For all the tree nodes falling in the same external node in $T_1$, we build a balanced tree from the array, and directly attach the root to the external node.
The work of scattering is $O(m'\log n')$, and the span is $O(\log m'+\log n')=O(\log n')$ (given $m'<\sqrt{n'+m'}$).\footnote{
In some models that allow multi-way forking (e.g., the PRAM model), the scatter step can be done effectively in optimal $O\left(m'\log\left(\frac{n'}{m'}+1\right)\right)$ work and $O(\log n')$ span, which makes the whole base case algorithm work-efficient.} In total, this algorithm has work $O(m'\log n')$ and span $O(\log n')$.

%After this step, the two trees have been merged into a larger tree, but imbalance might occur anywhere in the tree. We then use the algorithm in section \ref{sec:set:rebalance} (Algorithm \ref{algo:rebalance}) to rebalance the tree.

\para{Difference}($T_1-T_2$).
As shown in Figure \ref{fig:basecase}(c), we will only need to search in $T_1$ for each element in $T_2$, and for those appearing in both sets, we mark them as a tomb in $T_1$.
Then we call Algorithm~\ref{algo:rebalance} to filter out the tombs and rebalance the tree.
%The total cost excluding rebalancing is $O(m\log n)$ work and $O(\log n)$ span.
Similarly as the \union{} algorithm, applying Corollary \ref{cor:rebalancework} we can show that this algorithm has work $O(m'\log n')$ and span $O(\log n')$.

\para{Difference}($T_2-T_1$).
Similarly, as shown in Figure \ref{fig:basecase}(d), we search in $T_1$ for each element in $T_2$, and for those appearing in both sets, we mark them as a tomb in $T_2$.
Then we use Algorithm~\ref{algo:rebalance} to rebalance the tree.
The total cost is also $O(m'\log n')$ work and $O(\log n')$ span.

%% file: setset/proof.tex
%\subsection{Proofs for the Optimality of Our Set-set Algorithms}
\section{Cost Analysis for Set-set Operations}%Proofs for the Optimality of Our Set-set Algorithms}
\label{sec:set:proof}

We now prove the work and span of Algorithm~\ref{algo:setset}, which proves Theorem~\ref{thm:setset}.

\subsection{Preliminary and Useful Lemmas}
We start with presenting some useful definitions and Lemmas. We first recall the following two definitions we gave in Section \ref{sec:set:prelim}.

\begin{definition}
In a tree $T$, the \emph{\upper{s}} of an element $k \in U$, not necessarily in $T$, are all the nodes in $T$ on the search path
to $k$ (inclusive).
\end{definition}

\begin{definition}
In a tree $T$, an element $k \in U$ \emph{falls into} a subtree $T_x\in T$, if the search path to $k$ in $T$ overlaps the subtree $T_x$.
\end{definition}

We now present some useful lemmas and their proofs.

\hide{
\begin{lemma}
  When $n\ge m$, $1+\log_2 \frac{n}{m}$ is asymptotically equivalent to $\log_2 \left(\frac{n}{m}+1\right)$.  \guy{seems too trivial to be a lemma.}
\end{lemma}
\begin{proofsketch}
  When $n\ge 2m$ the stated lemma is straightforward. When $1\le n/m\le 2$, $1\le \log \frac{n}{m} +1\le 2$, and $1\le \log \left(\frac{n}{m}+1\right)\le \log_2 3$.
\end{proofsketch}
}

\begin{lemma}
\label{lem:uppers}
Suppose a tree $T$ of size $n$ satisfies that for any subtree $T_x\in T$, the height of $T_x$ is $O(\log |T_x|)$. Let $S$ be the set of all the \upper{} of $m\le n$ elements in $T$, then $|S|\in\bound$.
\end{lemma}
\begin{figure}
    \centering
    \includegraphics[width=0.6\columnwidth]{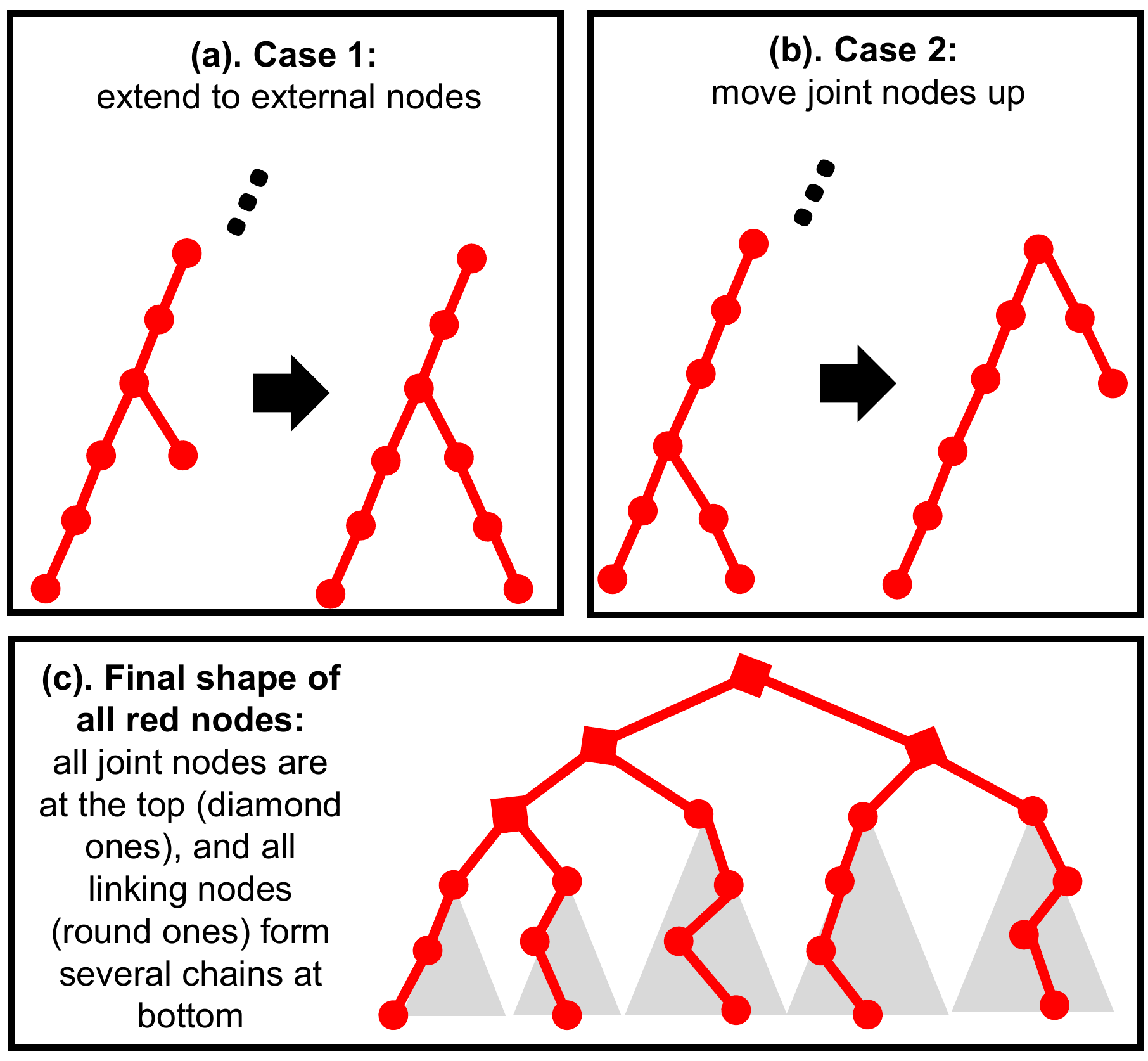}
    \caption{An illustration about adjusting red nodes in Proof \ref{lem:uppers}. (a) Extend all inner joint nodes to some external node. (b) Move all joint nodes to upper levels as far as possible. (c) The final shape after adjusting the red nodes.  All joint nodes are at the top, and all linking nodes form chains at bottom.}\label{fig:treepath}
\end{figure}
\begin{proof}
  First of all, we find all the searching paths to all the $m$ elements, and mark all related nodes on the path as red. All the red nodes form a tree structure, which is a connected component of $T$. We then adjust the red nodes, such that the number of red nodes does not decrease. We define a red node with two red children as a \emph{joint} node, and a red node with one red child as a \emph{linking} node.
  \begin{enumerate}
    \item First, as shown in Figure \ref{fig:treepath} (a), any of the red nodes are internal nodes in $T$, we arbitrarily extend it to some external node in the tree.
    \item Second, as shown in Figure \ref{fig:treepath} (b), if there is any linking node $v$ with some joint nodes as its descendent, then we move the first joint node of its descendants up to replace the non-red child of $v$.
  \end{enumerate}
  We repeat the two steps until there is no such situations. These adjustments only make the total number of red nodes larger. Finally we will have all joint nodes on the top levels of the tree, forming a connected component. All the linking nodes form several (at most $m$) chains at the bottom levels. The total number of joint nodes is $O(m)$ because the number of chains is at most $m$. For all the linking nodes, they form several chains. These chains starts from a linking node $u_i$ whose parent is a joint node $p(u_i)$. All such $p(u_i)$ nodes have disjoint subtrees. We assume the size of $p(u_i)$ is $n_i$, then we have $\sum_{i=1}^{m}n_i\le n$. The total length of all chains is:
  \begin{align*}
    &\sum_{i=1}^{m} \log (n_i+1)\le  m\log\left(\frac{\sum_{i=1}^{m}n_i+1}{m}\right)\le m\log\left(\frac{n}{m}+1\right)
  \end{align*}
  The above inequality can be shown by the Jensen's inequality.
\end{proof}
%The proof is given in the appendix.

\begin{lemma}
\label{lem:sumbound}
  Let $\sum_{i=1}^{k}m_i=m$,  $\sum_{i=1}^{k}n_i=n$ and $\forall i, m_i,n_i \in \mathbb{Z^+}, m_i<\sqrt{n_i+m_i}$,
  then we have $\sum_{i=1}^{k} m_i\log n_i \in \bound$.
\end{lemma}

\begin{proof}[Proof of Lemma~\ref{lem:sumbound}]
  We will show that given the conditions above, the maximum value of $\sum_{i=1}^{k} m_i\log n_i$ is no more than $c\boundcontent$ for some constant $c$.
  This follows the observation that $m_i\log n_i+m_j\log n_j$ gets the maximum value when $\frac{n_i}{n_j}=\frac{m_i}{m_j}$, given $n_i+n_j$ is fixed. Then for any two terms $n_i$ and $n_j$ on the left-hand side of the equation, if $n_i$ and $n_j$ are not distributed as the ratio of $m_i/m_j$, we re-distribute them as so, and the objective will never decrease. Thus the maximum value of the objective is when $n_i=\frac{m_i}{m}n$. This might invalidate the condition that $m_i<\sqrt{n_i+m_i}$, but will only make the bound looser.
  In this case, $m_i<\sqrt{n_i+m_i}<\sqrt{2n_i}=\sqrt{\frac{2m_i}{m}n}$, leading to $m_i<{2n\over m}$. Then we have:
    \begin{align*}
  \sum_{i=1}^{k} m_i\log n_i &\le \sum_{i=1}^{k} m_i\log \frac{m_i}{m}n\\
  &=\sum_{i=1}^{k} m_i\log \frac{n}{m} + \sum_{i=1}^{k} m_i\log {m_i}\\
  &<m\log \frac{n}{m}+\sum_{i=1}^{k} m_i\log {\frac{2n}{m}}\\
  &\in \bound
  \end{align*}

\end{proof}

\begin{lemma}
\label{lem:rotation}
  For a binary tree $T$, if its left and right subtrees are both valid WBB[$\alpha$] trees, and their weights differ by no more than $2/\alpha$, then we can rebalance $T$ as a WBB[$\alpha$] tree by a constant number of rotations, given $0 < \alpha \leq 1-1/\sqrt{2}$.
\end{lemma}

Lemma \ref{lem:rotation} is a special case of Lemma 3 in \cite{blelloch2016just} and can be proved in a similar way as shown in \cite{blelloch2016just}.

\subsection{The Proof for Theorem~\ref{thm:setset}}
% \begin{theorem} {\em(Optimal algorithm for ordered set algorithms)}
%   Algorithm~\ref{algo:setset} for \union{}, \intersection{} and \difference{} of two ordered sets of size $n$ and $m<n$ has $\bound$ work and $O(\log n)$ span in the \bfmodel{}.
% \end{theorem}
We now prove the cost bound of Algorithm~\ref{algo:setset}. Although Algorithm~\ref{algo:setset} is not very complicated, the analysis is reasonably involved, especially when we need to analyze all three different ordered set operations.
We first analyze the span in Lemma~\ref{lem:setdepth}, and the work bounds for the sketch step (Algorithm~\ref{algo:sketch}) in Lemma~\ref{lem:sketchingwork}, and the rebalance step (Algorithm~\ref{algo:rebalance}) in Lemma~\ref{lem:rebalancework}.
Combining these three lemmas proves the theorem.

Recall in the divide-and-conquer algorithm, we use \emph{pivots} (or \emph{splitters}\footnote{Throughout this section we use ``pivot'' and ``splitter'' interchangeably.}) to split the input trees to subproblems, and connect them back regardless of balance. We do so until reaching the base cases. As a result, all pivots in the algorithm form the upper levels of the sketch $T'$.  We refer to these upper levels containing all the pivots throughout the algorithm as the \emph{skeleton} of $T'$.
The skeleton consists of a connected collection of full binary trees of different sizes.
We first show two useful lemmas, revealing some useful information about the sketch tree $T'$.
We first bound the size of the skeleton of $T'$.

\begin{lemma}
\label{lem:mpivots}
For two input trees of sizes $m$ and $n\ge m$, there are in total $O(m)$ pivots in Algorithm \ref{algo:sketch}.
\end{lemma}
\begin{proof}
  We look at the tree skeleton consisting of all pivots in the algorithm. We will show that there at most $O(m)$ leaves in this skeleton. Each leaf corresponds to a function call of the base case. For any two sibling leafs, at least one of the elements in $T_S$ must fall into one of the two base cases, otherwise the algorithm should come to the base case on the upper level. This means that there can be at most $2m$ base cases.
\end{proof}

Intuitively, this is because when there are more than $m$ pivots, some of them must have an empty $T_2$ as input, which goes to the base case (Line \ref{line:basecaset1} in Algorithm \ref{algo:sketch}) immediately.

We then bound the height of all subtrees in the sketch tree $T'$.
\begin{lemma}{\em(Subtree height in sketch tree $T'$)}
\label{lem:step1height}
For any subtree in $T'$ obtained by the sketching step in Algorithm~\ref{algo:setset}, the height is no more than $O(\log n')$ for a subtree of size $n'$.
\end{lemma}
\begin{proof}
For any subtree $T_x \in T'$, there are two parts. The first several upper levels are in the skeleton, and the bottom several levels are those obtained by base cases. For the base case subtrees, they are balanced, and thus the height is at most $O(\log n')$. For the skeleton, it consists of several full binary trees of size $d$, for different $d$ values of recursion calls. Assume on the topmost (complete) level it is $d_0$-way dividing, so the first several levels should be a full binary tree of height $O(\log_2 d_0)$. Then for the next levels, it is at most $\sqrt{d_0}$-way, so the height is at most $\log_2 \sqrt{d_0}$. So on so forth. There also can be an incomplete full binary tree above $d_0$, and the height can be at most $\log_2 d_0$. Therefore the height of the skeleton is at most $2\log_2 d_0 + \log_2 \sqrt{d_0}+\dots=O(\log d_0)=O(\log n)$.
\end{proof}

This is guaranteed by the base case algorithms and the connecting substep of the \mf{Sketching} step.

We now prove the span of the algorithm.
Intuitively, the sketching step is a $\sqrt{n+m}$-way recursion, and thus requires logarithmic span.
For the rebalancing step, either reconstruction or rebalance by rotations costs a constant amortized time per level.
We formally show the following lemma.

\begin{lemma}
\label{lem:setdepth}
  Algorithm \ref{algo:setset} has span $O(\log n)$ in the \bfmodel{}.
\end{lemma}
\begin{proof}
  We first look at the \mf{Sketching} step working on two trees of size $n'$ (from the original large tree) and $m'$ (from the original small tree). The span for finding $d-1$ pivots, checking tombs, splitting, and the \texttt{connecting} functions are all $O(\log (n'+m'))$. For the next level of recursive call, the size of the problem shrink to $\sqrt{m'+n'}$. The span of a base case is also $O(\log (m'+n'))$. The recursion is therefore:
  $$T(m+n)=c_1 \log (m+n) + T(\sqrt{m+n})$$
  The solution is $T(m+n)\in O(\log (m+n))$.

  For the \mf{Rebalancing} step. We show this by induction that the span to rebalance a tree $T$ obtained by the sketching step is $O(\log |T|)$. The base case is straightforward. Other than the base case, there are two cases.

  If a subtree has its left and right children unbalanced, we need to reconstruct the tree structure. This requires to flatten all elements in an array and then build a complete binary tree on top of the array. The span of this step is no more than the height of the tree, which, according to Lemma \ref{lem:step1height}, is $O(\log |T|)$.

  Otherwise, each of the left and right subtrees has no more than $c|T|$ elements for some constant $c<1$. Rebalancing each of them, according to the inductive hypothesis, only need span $O(\log c|T|)$. There is an extra cost in rebalancing at the root of $T$ by rotation, which is constant time.

  For both cases, settling the subtree $T$ cost span $O(\log |T|)$. For our algorithm, it is $O(\log (m+n))=O(\log n)$.
\end{proof}

\hide{
\begin{lemma}
\label{lem:unionbalance}
  For any subtree $T_x$ in sketch $T'$, $T_x$ is balanced if tombs are not considered.
\end{lemma}
\begin{proof}
  The sketch is effectively the result of \union{} on the two input trees. If there are no duplicates in the two input sets, the \tchunk{s} are of the same size.  Considering the duplicates, the size of each \tchunk{} can shrink by at most a half. In this case, the tree is still balanced under the weight-balanced invariant.
\end{proof}
This lemma guarantees that the results returned by the if-condition at Line \ref{line:removelast} are valid.}

Next we prove the work of the algorithm, we start with showing that the total work of all bases cases is $\bound$.
\begin{lemma}
\label{lem:basecasework}
For two input trees of sizes $m$ and $n\ge m$, all base cases in Algorithm~\ref{algo:setset} require $\bound$ work.
\end{lemma}
\begin{proof}[Outline]
All the non-trivial base cases are when $m'<\sqrt{n'}$ in the recursive calls,
which will directly invoke the base case algorithms shown in Section \ref{sec:set:inefficient}. All such base case algorithms have $O(m'\log n')$ work.
Suppose there are $k$ such base case calls, each with input size $m_i$ and $n_i$ for $1\le i\le k$, where $\sum n_i=n, \sum m_i=m$ and $m_i<\sqrt{n_i+m_i}$. Then the total work is asymptotically $\sum_{i=1}^{k}m_i\log n_i$, which is $O(\bound)$ based on Lemma~\ref{lem:sumbound}.
\end{proof}

Next, we prove that the sketching step (Algorithm~\ref{algo:sketch}) uses work $\bound$.
This process is a leaf-dominated recursion and thus the total work is bounded by all bases cases.

\begin{lemma}
\label{lem:sketchingwork}
  The sketching step (Algorithm~\ref{algo:sketch}) has $\bound$ work.
\end{lemma}
\begin{proof}
  First of all, for all base cases, the total cost is $\bound$ based on Lemma \ref{lem:basecasework}.

  Excluding the base cases, we now look at the work caused by the \emph{sketching} step on two trees of size $n'$ (from the original large tree) and $m'$ (from the original small tree).

  In the parallel-for loops, the main work are for 1) finding the splitters, 2) splitting each tree into $d$ pieces and checking if each splitter is a tomb. All these cost work $O(\sqrt{n'+m'}\log (n'+m'))$. Then there will be $d\approx \sqrt{n'+m'}$ recursive calls, each with size about $\sqrt{n'+m'}$.

  For all connecting steps, since there are at most $O(m)$ base cases, there are at most $O(m)$ connecting pivots. Each connection costs a constant time, which means that this part only costs work $O(m)$.

%  We first prove it is true for base cases. For base cases, the cost is $O(m'\log n')$. Since $m'<\sqrt{m'+n'}<\sqrt{2n}$, which means that $(m')^2<2n'=2\frac{(n')^2}{n'}$. This implies $n'<\frac{n^2}{m^2}$.
%  Therefore, $m'\log n'<m'\log \frac{(n')^2}{(m')^2}=O(n'\log \frac{n'}{m'})=cn'\log \frac{n'}{m'}$ for some constant $c$.

%  We next prove the cost for an arbitrary node. As mentioned above, the total cost should be:
%
%  \begin{align*}
%    &\sum_{i=1}^{\sqrt{n'+m'}} \left(cm_i\log\frac{n_i}{m_i}+2\log(n_i+m_i)\right) + \sqrt{n'+m'}\log(n'+m') \\
%    = & \sum_{i=1}^{\sqrt{n'+m'}} \left(cm_i\log\frac{n_i}{m_i}\right)+\sum_{i=1}^{\sqrt{n'+m'}} \left(\log(\sqrt{n'+m'})\right)\\
%    +& \sqrt{n'+m'}\log(n'+m') \\
%    \le & cm'\log\frac{n'}{m'}+\sqrt{n'+m'}\log(n'+m') + \sqrt{n'+m'}\log(n'+m')\\
%    = & cm'\log\frac{n'}{m'}+2\sqrt{n'+m'}\log(n'+m')
%  \end{align*}

  We then prove that the total work in parallel for loops is also $O\left(m'\log\frac{n'}{m'}\right)$. Consider the recurrence tree of this algorithm, which has $\sqrt{m'+n'}$-way fan-out for a recursion with sizes $n'$ and $m'$, and the cost of the current node is $\sqrt{n'+m'}\log(n'+m')$. In our algorithm, we stop recursing when $m'\ge \sqrt{m'+n'}$ and call the base case algorithm instead of waiting until $m'$ or $n'$ reaches zero. As a result, the tree is not perfectly balanced. Some branches can be shallow (reaching base cases earlier) and some can be deep (reaching base cases later).
  The cost at each node, however, is exactly the same in the same level, but decreases as the tree goes deeper. Therefore, given the number of leaves fixed, the worst case occurs when all nodes are at the topmost several levels.
  Let $t=m'+n'$, the total work for these operations is:
  \begin{align*}
    &t^{1/2}\log t + t^{1/2}\cdot t^{1/4}\log t^{1/2} + t^{1/2}\cdot t^{1/4}\cdot t^{1/8}\log t^{1/4} + \dots\\
    =&\sum \frac{\log t}{2^i}t^{1-1/2^{i+1}}
  \end{align*}
  This recurrence is leaf-dominated. %The total work at all \emph{leaves} is $O(t)$.
  The recurrence suggested that in the $i$-th term, there are $t^{1-1/2^{i+1}}$ subtasks each costing $\frac{\log t}{2^i}$ work.
  We then want to know, when the recurrence stops, how much work we pay. From Lemma \ref{lem:mpivots}, we know the total number of pivots is no more than $m'$, and thus $t^{1-1/2^{i+1}}\le m'$.
  Let $x=t^{1/2^{i+1}}$, the total leaf cost is asymptotically
  $$t^{1-1/2^{i+1}}\log t^{1/2^{i+1}}=\frac{t}{x}\log x$$
  Given $t$ fixed, this function is decreasing when $x\rightarrow +\infty$. Considering $t^{1-1/2^{i+1}}\le m'$, we know that $x\ge t/m'$. Plug in $x=t/m'$ we know that the total work is at most
  $$O\left(\frac{t}{x}\log x\right)=O\left(m'\log\frac{t}{m'}\right)=O\left(m'\log\frac{n'+m'}{m'}\right)=O\left(m'\log\frac{n'}{m'}\right)$$

  %The total cost is asymptotically bounded by all base cases.

  Therefore, the work of the sketching step is $\bound$.
\end{proof}

Finally, we show the total work of rebalancing is $\bound$. Intuitively, this is because the amortized time to settle a tree node is a constant for either rotation or rebalancing. In addition, for filling in all tombs, the total work is $\bound$. We formally prove it as follows.
\begin{lemma}
\label{lem:rebalancework}
  The rebalancing step (Algorithm~\ref{algo:rebalance}) has $\bound$ work.
\end{lemma}
\begin{proof}
  We first show that for \union{} we do not need to rebalance. If there are no duplicates in the two input sets, all the \tchunk{s} are of the same size.  Considering the duplicates, the size of each \tchunk{} can shrink by at most a half.   In this case, and since the pivots perfectly balance not considering duplicates, the tree is still balanced under the weight-balanced invariant.

% \guy{Is this true?   It seems like it will be near balanced, but not completely balanced?}

%   \yihan{Yes. This is only when we are interested in \emph{union} and there are \emph{no} duplicates. In these cases, there are no tombs. Considering that our splitting criteria always find the global $\sqrt{m+n}, 2\sqrt{m+n}, \dots$ element, which guarantees the tree on the top levels are perfectly balanced. The balance of
%   the lower level small subtrees are guaranteed by the base case algorithms, which
%   will call the rebalancing algorithm at the end.}

  We then consider \intersection{} and \difference{}. In Algorithm \ref{algo:rebalance}, the total work consists of three parts: filling up the tombs (Line \ref{line:removelast1} and \ref{line:removelast2}), reconstruction (Line \ref{line:reconstruct}), and rotation (Line \ref{line:rotation}). We note them as $W_1$, $W_2$ and $W_3$, respectively.

  We first prove that $W_2+W_3$ is $\bound$. We ignore the cost of \func{RemoveLast} for now. Note that all subtrees obtained by base cases are balanced, so the rebalancing process will not touch those parts in $T'$ (Line \ref{line:skip1} in Algorithm \ref{algo:rebalance}), and will only visit the pivots in the skeleton. %This is guaranteed because in the results of the bases cases, there will be no tombs, so the algorithm will return at Line \ref{line:removelast}. There are at most $O(m)$ pivots on the skeleton. We then discuss in two cases.

  %First, in \difference{} ($T_L-T_S$), there are $m$ tombs in $T'$, corresponding to the $m$ elements in $T_S$. For each tomb, we mark all of its \upper{s} in $T'$ as red.

  %Second, in \difference{} ($T_S-T_L$) and \intersection{}, there are at most $m$ elements that are not tombs, corresponding to the $m$ elements in $T_S$. For each non-tomb, we mark all of its \upper{s} in $T'$ as red.

  %Define $M$ as the number of elements from $T_S$ that also appear in the result tree.
  For all elements $k\in T_S$, we mark all their \upper{s} in $T'$ as red.
  If there is no element $k\in T_S$ falling into a subtree in $T'$, this subtree will be skipped over directly (Line \ref{line:skip1}). Therefore the red nodes are the only nodes visited by the algorithm for reconstruction and rotation.

  Based on Lemma \ref{lem:uppers} and Lemma \ref{lem:step1height}, there are at most $\bound$ such red nodes in the skeleton. We denote the number of red nodes in a subtree $T$ (or a subtree rooted at $v$) as $R(T)$ (or $R(v)$). We first show that if we do not consider the cost of filling the tombs, the rebalancing cost for any subtree $T_x\in T'$ is asymptotically no more than $R(T_x)$.  We show this by induction.

  First, only those red nodes will be reached in the rebalancing step. Therefore the base case holds (e.g., when the tree is just one red node). For each red node $v$ in $T'$, there are two cases.
  \begin{enumerate}
    \item $v$'s left and right subtrees are almost balanced. In this case $v$ will be settled by a constant number of rotations. Considering the inductive hypothesis, the total work is asymptotically no more than the number of red nodes in the whole subtree.
    \item The sizes of $v$'s left and right subtrees differ by more than a constant factor $c$. Then we need to reconstruct the subtree, and the work is linear to the size of $v$'s subtree. We use $M$ to represent the number of elements in $T_S$ that fall into $v$'s subtree.

        We first show $M\le R(v)$. This is because the red nodes are always ancestors (inclusive) of those tree nodes from $T_S$. Therefore, $M\le R(v)$ holds for all tree nodes $v$.

        Recall that all the \tchunk{s} are designed to be of the same size. For \union{} (i.e., no tombs in the tree), the size of two \tchunk{s} can differ by at most a factor of two (due to duplicates), which will not cause imbalance. For \difference{} and \intersection{}, the only reason for the imbalance is that there are tombs being removed. There are two cases. We next show that in either case, the difference between the left subtree and right subtree of $v$ is no more than $M$.
        \begin{enumerate}
          \item In \difference{} ($T_L-T_S$). All tombs are elements in $T_S$. In these case, the difference of the left subtree and right subtree of $v$ is no more than $M$.
          \item In \difference{} ($T_S-T_L$) and \intersection{}. In these two cases all elements in $v$'s subtree must appear in $T_S$. Thus the size of the whole subtree (excluding tombs) rooted at $v$ is no more than $M$, and therefore the difference of the left subtree and right subtree of $v$ is no more than $M$.
        \end{enumerate}
        Recall that we use $l(v)$ and $r(v)$ to denote the left and right subtrees of a node $v$, respectively. Therefore the above statements proves that $c \cdot \text{size}(v) \le |\text{size}(r(v))-\text{size}(l(v))|\le M \le R(v)$, for some constant $c$. Therefore, the total work $O(\text{size}(v))$ is also $O(R(v))$.
  \end{enumerate}
  Therefore we proved that in any of the subtree $T_x\in T'$ the work is asymptotically no more than $R(T_x)$.  This will also be true for $T'$ itself, and $W_2+W_3$ for whole $T'$ is $\bound$.

  Next we show that $W_1$ is $\bound$. For $W_3$, the algorithm will pop up at most $m$ elements to fill in the tombs. Each such operation follows the right spine in the corresponding subtree of size $n'$ in $O(\log n')$ time. All such subtrees are disjoint. Therefore, based on Lemma \ref{lem:sumbound}, the total cost is $\bound$.
\end{proof}

Combining Lemmas \ref{lem:basecasework}, \ref{lem:sketchingwork} and \ref{lem:rebalancework} gives the following Lemma about the work of Algorithm \ref{algo:setset}.

\begin{lemma}
\label{lem:setwork}
The total work of Algorithm \ref{algo:setset} is $\bound$.
\end{lemma}

Lemma \ref{lem:rebalancework} also indicates the following corollary:
\begin{corollary} \label{cor:rebalancework}
\mf{Rebalance}$(T,\texttt{false})$ as shown in Algorithm \ref{algo:rebalance} on a weight-balanced tree $T$ of size $n$ with $m$ tombs in it cost work $\bound$. For this case we ignore the part of if-condition about base case at line \ref{line:skip1}.
\end{corollary}

This corollary is weaker than the condition considered in Lemma \ref{lem:rebalancework}, and thus can be shown using a similar proof. This corollary can be used to bound the work of the base case and the work for \union{} and \difference{} (see Section \ref{sec:set:inefficient}).
Combining Lemmas \ref{lem:setdepth}--\ref{lem:rebalancework} gives
Theorem~\ref{thm:setset}.

%% file: app-others.tex
\section{Proof of Theorem \ref{thm:complexity}}
\label{app:complexity}

%\guy{proofs can go to the appendix.   Also probably an easier proof
%  for the second half is to simulate the BF model on the PRAM.}
\begin{proof}
(Outline).   The first inclusion is well known~\cite{Borodin77}.
% Simulating an $\NC^1$ circuit in $\BF^1$ (without \texttt{TS}) is
% straightforward.  We start at the output bit and follow backwards
% forking two child threads for each of the two inputs.  Since the
% circuit is $O(\log n)$ deep, we will only fork off polynomially many
% threads, limiting the work to $O(n^k)$.
To simulate $L$ in $\BF^1$ we can construct a graph representing the
state transition diagram of the logspace computation.  Each state
looks at one bit of the input tape and has two edges out, one for 0
and one for 1.  The construction is logspace uniform.  For a
particular input of length $n$ we select the appropriate edges, which
forms a forest with the accept states appearing as some of the roots.
An Euler Tour can be built on the trees in the forest---not hard to do
in logarithmic span.  Using deterministic list-contraction, as
described at the end of the next section, the simulation can identify
if the start state has an accept state as its root.  The full
computation has $O(\log n)$ span and does polynomial work (there are
polynomially many states).

To simulate $\BF^1$ in $\AC^1$ we could use
know results for simulating the CRCW PRAM in circuit depth
proportional to time~\cite{ChandraSV84}, but this would require randomization
since our simulation of the \bfmodel{} on the PRAM uses
randomization.  However, it is not difficult to extend the ideas for a
deterministic simulation.  The idea is that the shared memory can be
simulated in constant depth per step~\cite{KarpR90}, as can the processor.
The processor simulation can take advantage of the fact that our registers only
have $O(\log n)$ bits and instructions only take a constant number of
registers.  This allows instructions to be simulated in constant
depth and polynomial size, effectively by table lookup (just index the
correct solution based on the $O(\log n)$ input bits).  The forking
can be simulated by assuming that every instruction might be a fork,
so create circuits for two copies of the processor as the children of each instruction.  Since the
computation has $O(\log n)$ span, at most a polynomial number of
processors simulators are required in the circuit.

% The simulation of $\PRAM_{CRCW}^1$ in $\TRAM^1$ has already been
% discussed, and the final can be done by simulating the  $\TRAM^1$ on a
% $\PRAM_{CRCW}^2$~\cite{BGM99} which is equivalent to $\NC^2$.
\end{proof}

% It is interesting, although perhaps not surprising, that the power of
% unbounded forking in the \tram{} is reflected in the power of circuits
% with elements with unbounded fan-in.    We note none of these
% simulations are work efficient. 

%% file: main.bbl
\begin{thebibliography}{10}

\bibitem{Acar02}
Umut~A. Acar, Guy~E. Blelloch, and Robert~D. Blumofe.
\newblock The data locality of work stealing.
\newblock {\em Theoretical Computer Science (TCS)}, 35(3), 2002.

\bibitem{ACGRS18}
Umut~A. Acar, Arthur Chargu{\'e}raud, Adrien Guatto, Mike Rainey, and Filip
  Sieczkowski.
\newblock Heartbeat scheduling: Provable efficiency for nested parallelism.
\newblock In {\em ACM Conference on Programming Language Design and
  Implementation (PLDI)}, pages 769--782, 2018.

\bibitem{agrawal2014batching}
Kunal Agrawal, Jeremy~T. Fineman, Kefu Lu, Brendan Sheridan, Jim Sukha, and
  Robert Utterback.
\newblock Provably good scheduling for parallel programs that use data
  structures through implicit batching.
\newblock In {\em {ACM} Symposium on Parallelism in Algorithms and
  Architectures (SPAA)}, 2014.

\bibitem{workingset}
Kunal Agrawal, Seth Gilbert, and Wei~Quan Lim.
\newblock Parallel working-set search structures.
\newblock In {\em {ACM} Symposium on Parallelism in Algorithms and
  Architectures (SPAA)}, 2018.

\bibitem{akhremtsev2016fast}
Yaroslav Akhremtsev and Peter Sanders.
\newblock Fast parallel operations on search trees.
\newblock In {\em {IEEE} International Conference on High Performance Computing
  (HiPC)}, 2016.

\bibitem{Alonso1996}
Laurent Alonso and Ren Schott.
\newblock A parallel algorithm for the generation of a permutation and
  applications.
\newblock {\em Theoretical Computer Science (TCS)}, 159(1), 1996.

\bibitem{alstrup2004nearest}
Stephen Alstrup, Cyril Gavoille, Haim Kaplan, and Theis Rauhe.
\newblock Nearest common ancestors: A survey and a new algorithm for a
  distributed environment.
\newblock {\em Theory of Computing Systems (TOCS)}, 37(3):441--456, 2004.

\bibitem{Anderson1990}
Richard~J. Anderson.
\newblock Parallel algorithms for generating random permutations on a shared
  memory machine.
\newblock In {\em {ACM} Symposium on Parallelism in Algorithms and
  Architectures (SPAA)}, 1990.

\bibitem{AndersonMiller1990}
Richard~J Anderson and Gary~L Miller.
\newblock A simple randomized parallel algorithm for list-ranking.
\newblock {\em Information Processing Letters}, 33(5):269--273, 1990.

\bibitem{arge2013dynamic}
Lars Arge, Johannes Fischer, Peter Sanders, and Nodari Sitchinava.
\newblock On (dynamic) range minimum queries in external memory.
\newblock In {\em Workshop on Algorithms and Data Structures}, pages 37--48.
  Springer, 2013.

\bibitem{ABP01}
N.~S. Arora, R.~D. Blumofe, and C.~G. Plaxton.
\newblock Thread scheduling for multiprogrammed multiprocessors.
\newblock {\em Theory of Computing Systems (TOCS)}, 34(2), Apr 2001.

\bibitem{Baase93}
Sara Baase.
\newblock Introduction to parallel connectivity, list ranking, and {Euler} tour
  techniques.
\newblock In John Reif, editor, {\em Synthesis of Parallel Algorithms}, pages
  61--114. Morgan Kaufmann, 1993.

\bibitem{BBFGGMS16}
Naama Ben-David, Guy~E. Blelloch, Jeremy~T. Fineman, Phillip~B. Gibbons, Yan
  Gu, Charles McGuffey, and Julian Shun.
\newblock Parallel algorithms for asymmetric read-write costs.
\newblock In {\em {ACM} Symposium on Parallelism in Algorithms and
  Architectures (SPAA)}, 2016.

\bibitem{BBFGGMS18}
Naama Ben-David, Guy~E. Blelloch, Jeremy~T Fineman, Phillip~B Gibbons, Yan Gu,
  Charles McGuffey, and Julian Shun.
\newblock Implicit decomposition for write-efficient connectivity algorithms.
\newblock In {\em {IEEE} International Parallel and Distributed Processing
  Symposium (IPDPS)}, 2018.

\bibitem{bender2019small}
Michael~A. Bender, Alex Conway, Mart{\'\i}n Farach-Colton, William Jannen,
  Yizheng Jiao, Rob Johnson, Eric Knorr, Sara McAllister, Nirjhar Mukherjee,
  Prashant Pandey, et~al.
\newblock Small refinements to the {DAM} can have big consequences for
  data-structure design.
\newblock In {\em {ACM} Symposium on Parallelism in Algorithms and
  Architectures (SPAA)}, pages 265--274, 2019.

\bibitem{bender2000lca}
Michael~A. Bender and Martin Farach-Colton.
\newblock The lca problem revisited.
\newblock In {\em Latin American Symposium on Theoretical Informatics (LATIN)},
  pages 88--94. Springer, 2000.

\bibitem{berkman1993recursive}
Omer Berkman and Uzi Vishkin.
\newblock Recursive star-tree parallel data structure.
\newblock {\em {SIAM} J. Scientific Computing}, 22(2):221--242, 1993.

\bibitem{blelloch1990pre}
Guy~E. Blelloch.
\newblock Prefix sums and their applications.
\newblock In John Reif, editor, {\em Synthesis of Parallel Algorithms}. Morgan
  Kaufmann, 1993.

\bibitem{Blelloch96}
Guy~E. Blelloch.
\newblock Programming parallel algorithms.
\newblock {\em Commun. {ACM}}, 39(3), March 1996.

\bibitem{BCGRCK08}
Guy~E. Blelloch, Rezaul~Alam Chowdhury, Phillip~B. Gibbons, Vijaya
  Ramachandran, Shimin Chen, and Michael Kozuch.
\newblock Provably good multicore cache performance for divide-and-conquer
  algorithms.
\newblock In {\em {ACM-SIAM} Symposium on Discrete Algorithms (SODA)}, 2008.

\bibitem{blelloch2016just}
Guy~E. Blelloch, Daniel Ferizovic, and Yihan Sun.
\newblock Just join for parallel ordered sets.
\newblock In {\em {ACM} Symposium on Parallelism in Algorithms and
  Architectures (SPAA)}, 2016.

\bibitem{blelloch2012internally}
Guy~E. Blelloch, Jeremy~T. Fineman, Phillip~B. Gibbons, and Julian Shun.
\newblock Internally deterministic parallel algorithms can be fast.
\newblock In {\em {ACM} Symposium on Principles and Practice of Parallel
  Programming (PPOPP)}, 2012.

\bibitem{BlellochFiGi11}
Guy~E. Blelloch, Jeremy~T. Fineman, Phillip~B. Gibbons, and Harsha~Vardhan
  Simhadri.
\newblock Scheduling irregular parallel computations on hierarchical caches.
\newblock In {\em {ACM} Symposium on Parallelism in Algorithms and
  Architectures (SPAA)}, 2011.

\bibitem{BG04}
Guy~E. Blelloch and Phillip~B. Gibbons.
\newblock Effectively sharing a cache among threads.
\newblock In {\em {ACM} Symposium on Parallelism in Algorithms and
  Architectures (SPAA)}, 2004.

\bibitem{BGM99}
Guy~E. Blelloch, Phillip~B. Gibbons, and Yossi Matias.
\newblock Provably efficient scheduling for languages with fine-grained
  parallelism.
\newblock {\em J. {ACM}}, 46(2), March 1999.

\bibitem{blelloch2010low}
Guy~E. Blelloch, Phillip~B. Gibbons, and Harsha~Vardhan Simhadri.
\newblock Low depth cache-oblivious algorithms.
\newblock In {\em {ACM} Symposium on Parallelism in Algorithms and
  Architectures (SPAA)}, 2010.

\bibitem{BGSS16}
Guy~E. Blelloch, Yan Gu, Julian Shun, and Yihan Sun.
\newblock Parallelism in randomized incremental algorithms.
\newblock In {\em {ACM} Symposium on Parallelism in Algorithms and
  Architectures (SPAA)}, 2016.

\bibitem{BGSS18}
Guy~E. Blelloch, Yan Gu, Julian Shun, and Yihan Sun.
\newblock Parallel write-efficient algorithms and data structures for
  computational geometry.
\newblock In {\em {ACM} Symposium on Parallelism in Algorithms and
  Architectures (SPAA)}, 2018.

\bibitem{blelloch2020randomized}
Guy~E. Blelloch, Yan Gu, Julian Shun, and Yihan Sun.
\newblock Randomized incremental convex hull is highly parallel.
\newblock In {\em {ACM} Symposium on Parallelism in Algorithms and
  Architectures (SPAA)}, 2020.

\bibitem{Blelloch1998}
Guy~E. Blelloch and Margaret Reid-Miller.
\newblock Fast set operations using treaps.
\newblock In {\em {ACM} Symposium on Parallelism in Algorithms and
  Architectures (SPAA)}, 1998.

\bibitem{blelloch1999pipelining}
Guy~E. Blelloch and Margaret Reid-Miller.
\newblock Pipelining with futures.
\newblock {\em Theory of Computing Systems (TOCS)}, 32(3), 1999.

\bibitem{BST12}
Guy~E. Blelloch, Harsha~Vardhan Simhadri, and Kanat Tangwongsan.
\newblock Parallel and {I/O} efficient set covering algorithms.
\newblock In {\em {ACM} Symposium on Parallelism in Algorithms and
  Architectures (SPAA)}, 2012.

\bibitem{BL98}
Robert~D. Blumofe and Charles~E. Leiserson.
\newblock Space-efficient scheduling of multithreaded computations.
\newblock {\em {SIAM} J. Scientific Computing}, 27(1), 1998.

\bibitem{blumofe1999scheduling}
Robert~D. Blumofe and Charles~E. Leiserson.
\newblock Scheduling multithreaded computations by work stealing.
\newblock {\em J. {ACM}}, 46(5):720--748, 1999.

\bibitem{Borodin77}
Allan Borodin.
\newblock On relating time and space to size and depth.
\newblock {\em {SIAM} J. Scientific Computing}, 6(4):733--744, 1977.

\bibitem{braginsky2012lock}
Anastasia Braginsky and Erez Petrank.
\newblock A lock-free b+ tree.
\newblock In {\em {ACM} Symposium on Parallelism in Algorithms and
  Architectures (SPAA)}, pages 58--67, 2012.

\bibitem{brownT80}
Mark~R. Brown and Robert~Endre Tarjan.
\newblock Design and analysis of a data structure for representing sorted
  lists.
\newblock {\em {SIAM} J. Scientific Computing}, 9(3):594--614, 1980.

\bibitem{budimlic2011design}
Zoran Budimli{\'c}, Vincent Cav{\'e}, Raghavan Raman, Jun Shirako, Sa{\u{g}}nak
  Ta{\c{s}}{\i}rlar, Jisheng Zhao, and Vivek Sarkar.
\newblock The design and implementation of the habanero-java parallel
  programming language.
\newblock In {\em Symposium on Object-oriented Programming, Systems, Languages
  and Applications (OOPSLA)}, pages 185--186, 2011.

\bibitem{ChandraSV84}
Ashok~K. Chandra, Larry~J. Stockmeyer, and Uzi Vishkin.
\newblock Constant depth reducibility.
\newblock {\em {SIAM} J. Comput.}, 13(2):423--439, 1984.

\bibitem{charles2005x10}
Philippe Charles, Christian Grothoff, Vijay Saraswat, Christopher Donawa, Allan
  Kielstra, Kemal Ebcioglu, Christoph Von~Praun, and Vivek Sarkar.
\newblock X10: an object-oriented approach to non-uniform cluster computing.
\newblock In {\em ACM SIGPLAN Conference on Object-oriented Programming,
  Systems, Languages, and Applications (OOPSLA)}, pages 519--538, 2005.

\bibitem{chowdhury2017provably}
Rezaul Chowdhury, Pramod Ganapathi, Yuan Tang, and Jesmin~Jahan Tithi.
\newblock Provably efficient scheduling of cache-oblivious wavefront
  algorithms.
\newblock In {\em {ACM} Symposium on Parallelism in Algorithms and
  Architectures (SPAA)}, pages 339--350, 2017.

\bibitem{CRSB13}
Rezaul~A. Chowdhury, Vijaya Ramachandran, Francesco Silvestri, and Brandon
  Blakeley.
\newblock Oblivious algorithms for multicores and networks of processors.
\newblock {\em Journal of Parallel and Distributed Computing}, 73(7):911--925,
  2013.

\bibitem{Cole1988}
Richard Cole.
\newblock Parallel merge sort.
\newblock {\em {SIAM} J. Scientific Computing}, 17(4), 1988.

\bibitem{CR17b}
Richard Cole and Vijaya Ramachandran.
\newblock Bounding cache miss costs of multithreaded computations under general
  schedulers: Extended abstract.
\newblock In {\em {ACM} Symposium on Parallelism in Algorithms and
  Architectures (SPAA)}, 2017.

\bibitem{Cole17}
Richard Cole and Vijaya Ramachandran.
\newblock Resource oblivious sorting on multicores.
\newblock {\em {ACM} Transactions on Parallel Computing (TOPC)}, 3(4), 2017.

\bibitem{Cole85}
Richard Cole and Uzi Vishkin.
\newblock Deterministic coin tossing and accelerating cascades: Micro and macro
  techniques for designing parallel algorithms.
\newblock In {\em {ACM} Symposium on Theory of Computing (STOC)}, pages
  206--219, 1986.

\bibitem{Cole1988ApproximatePS}
Richard Cole and Uzi Vishkin.
\newblock Approximate parallel scheduling. part i: The basic technique with
  applications to optimal parallel list ranking in logarithmic time.
\newblock {\em {SIAM} J. Scientific Computing}, 17:128--142, 1988.

\bibitem{cole1989faster}
Richard Cole and Uzi Vishkin.
\newblock Faster optimal parallel prefix sums and list ranking.
\newblock {\em Information and computation}, 81(3):334--352, 1989.

\bibitem{CongB05}
Guojing Cong and David~A. Bader.
\newblock An empirical analysis of parallel random permutation algorithms on
  {SMPs}.
\newblock In {\em Parallel and Distributed Computing and Systems (PDCS)}, 2005.

\bibitem{CLRS}
Thomas~H. Cormen, Charles~E. Leiserson, Ronald~L. Rivest, and Clifford Stein.
\newblock {\em Introduction to Algorithms (3rd edition)}.
\newblock MIT Press, 2009.

\bibitem{Czumaj96}
Artur Czumaj, Przemyslawa Kanarek, Miroslaw Kutylowski, and Krzysztof Lorys.
\newblock Fast generation of random permutations via networks simulation.
\newblock {\em Algorithmica}, 1998.

\bibitem{Dhulipala2018}
Laxman Dhulipala, Guy~E. Blelloch, and Julian Shun.
\newblock Theoretically efficient parallel graph algorithms can be fast and
  scalable.
\newblock In {\em {ACM} Symposium on Parallelism in Algorithms and
  Architectures (SPAA)}, 2018.

\bibitem{dhulipala2020semi}
Laxman Dhulipala, Charlie McGuffey, Hongbo Kang, Yan Gu, Guy~E Blelloch,
  Phillip~B Gibbons, and Julian Shun.
\newblock Semi-asymmetric parallel graph algorithms for nvrams.
\newblock {\em Proceedings of the VLDB Endowment (PVLDB)}, 13(9), 2020.

\bibitem{dinh2016extending}
David Dinh, Harsha~Vardhan Simhadri, and Yuan Tang.
\newblock Extending the nested parallel model to the nested dataflow model with
  provably efficient schedulers.
\newblock In {\em {ACM} Symposium on Parallelism in Algorithms and
  Architectures (SPAA)}, 2016.

\bibitem{persistence}
James~R. Driscoll, Neil Sarnak, Daniel~D. Sleator, and Robert~E. Tarjan.
\newblock Making data structures persistent.
\newblock {\em J. Computer and System Sciences}, 38(1):86--124, 1989.

\bibitem{Durstenfeld64}
Richard Durstenfeld.
\newblock Algorithm 235: Random permutation.
\newblock {\em Commun. {ACM}}, 7(7):420, 1964.

\bibitem{fatourou2019persistent}
Panagiota Fatourou, Elias Papavasileiou, and Eric Ruppert.
\newblock Persistent non-blocking binary search trees supporting wait-free
  range queries.
\newblock In {\em {ACM} Symposium on Parallelism in Algorithms and
  Architectures (SPAA)}, pages 275--286, 2019.

\bibitem{fischer2006theoretical}
Johannes Fischer and Volker Heun.
\newblock Theoretical and practical improvements on the {RMQ}-problem, with
  applications to {LCA} and {LCE}.
\newblock In {\em Symposium on Combinatorial Pattern Matching (CPM)}, pages
  36--48. Springer, 2006.

\bibitem{Frazer70}
W~Donald Frazer and Archie~C McKellar.
\newblock Samplesort: A sampling approach to minimal storage tree sorting.
\newblock {\em J. {ACM}}, 17(3):496--507, 1970.

\bibitem{frigo1998implementation}
Matteo Frigo, Charles~E. Leiserson, and Keith~H. Randall.
\newblock The implementation of the cilk-5 multithreaded language.
\newblock {\em ACM SIGPLAN Conference on Programming Language Design and
  Implementation (PLDI)}, 33(5):212--223, 1998.

\bibitem{Gibbons1996}
Phillip~B. Gibbons, Yossi Matias, and Vijaya Ramachandran.
\newblock Efficient low-contention parallel algorithms.
\newblock {\em J. Computer and System Sciences}, 53(3), 1996.

\bibitem{Gil91}
Joseph Gil.
\newblock Fast load balancing on a {PRAM}.
\newblock In {\em {IEEE} International Parallel and Distributed Processing
  Symposium (IPDPS)}, 1991.

\bibitem{Gil91a}
Joseph Gil, Yossi Matias, and Uzi Vishkin.
\newblock Towards a theory of nearly constant time parallel algorithms.
\newblock In {\em {IEEE} Symposium on Foundations of Computer Science (FOCS)},
  1991.

\bibitem{gilbert2019parallel}
Seth Gilbert and Wei~Quan Lim.
\newblock Parallel finger search structures.
\newblock {\em arXiv preprint arXiv:1908.02741}, 2019.

\bibitem{gu2015top}
Yan Gu, Julian Shun, Yihan Sun, and Guy~E Blelloch.
\newblock A top-down parallel semisort.
\newblock In {\em {ACM} Symposium on Parallelism in Algorithms and
  Architectures (SPAA)}, 2015.

\bibitem{Gustedt03}
Jens Gustedt.
\newblock Randomized permutations in a coarse grained parallel environment.
\newblock In {\em {ACM} Symposium on Parallelism in Algorithms and
  Architectures (SPAA)}, 2003.

\bibitem{Gustedt2008}
Jens Gustedt.
\newblock Engineering parallel in-place random generation of integer
  permutations.
\newblock In {\em Workshop on Experimental Algorithmics}, 2008.

\bibitem{Hagerup91}
Torben Hagerup.
\newblock Fast parallel generation of random permutations.
\newblock In {\em Intl. Colloq. on Automata, Languages and Programming
  {(ICALP)}}. Springer, 1991.

\bibitem{Herlihy91}
Maurice Herlihy.
\newblock Wait-free synchronization.
\newblock {\em ACM Transactions on Programming Languages and Systems (TOPLAS)},
  13(1):124–149, 1991.

\bibitem{hwang1972simple}
Frank~K. Hwang and Shen Lin.
\newblock A simple algorithm for merging two disjoint linearly ordered sets.
\newblock {\em {SIAM} J. Scientific Computing}, 1(1):31--39, 1972.

\bibitem{TBB}
https://www.threadingbuildingblocks.org.

\bibitem{jacob2014complexity}
Riko Jacob, Tobias Lieber, and Nodari Sitchinava.
\newblock On the complexity of list ranking in the parallel external memory
  model.
\newblock In {\em International Symposium on Mathematical Foundations of
  Computer Science}, pages 384--395. Springer, 2014.

\bibitem{JaJa92}
J.~JaJa.
\newblock {\em Introduction to Parallel Algorithms}.
\newblock Addison-Wesley Professional, 1992.

\bibitem{Java-fork-join}
http://docs.oracle.com/javase/tutorial/essential/concurrency/forkjoin.html.

\bibitem{KarpR90}
Richard~M. Karp and Vijaya Ramachandran.
\newblock Parallel algorithms for shared-memory machines.
\newblock In {\em Handbook of Theoretical Computer Science, Volume A:
  Algorithms and Complexity (A)}. MIT Press, 1990.

\bibitem{Knuth69}
Donald~E. Knuth.
\newblock {\em The Art of Computer Programming, Volume II: Seminumerical
  Algorithms}.
\newblock Addison-Wesley, 1969.

\bibitem{Miller85}
Gary~L. Miller and John~H. Reif.
\newblock Parallel tree contraction and its application.
\newblock In {\em {IEEE} Symposium on Foundations of Computer Science (FOCS)},
  pages 478--489, 1985.

\bibitem{milman2018bq}
Gal Milman, Alex Kogan, Yossi Lev, Victor Luchangco, and Erez Petrank.
\newblock Bq: A lock-free queue with batching.
\newblock In {\em {ACM} Symposium on Parallelism in Algorithms and
  Architectures (SPAA)}, pages 99--109, 2018.

\bibitem{nievergelt1973binary}
J{\"u}rg Nievergelt and Edward~M Reingold.
\newblock Binary search trees of bounded balance.
\newblock {\em {SIAM} J. Scientific Computing}, 2(1), 1973.

\bibitem{PP01}
Heejin Park and Kunsoo Park.
\newblock Parallel algorithms for red-black trees.
\newblock {\em Theoretical Computer Science (TCS)}, 262(1--2):415--435, 2001.

\bibitem{PVW83}
Wolfgang~J. Paul, Uzi Vishkin, and Hubert Wagener.
\newblock Parallel dictionaries in 2-3 trees.
\newblock In {\em Intl. Colloq. on Automata, Languages and Programming
  {(ICALP)}}, pages 597--609, 1983.

\bibitem{RR89}
Sanguthevar Rajasekaran and John~H. Reif.
\newblock Optimal and sublogarithmic time randomized parallel sorting
  algorithms.
\newblock {\em {SIAM} J. Scientific Computing}, 18(3), 1989.

\bibitem{Randade98}
Abhiram Ranade.
\newblock A simple optimal list ranking algorithm.
\newblock In {\em {IEEE} International Conference on High Performance Computing
  (HiPC)}, 1998.

\bibitem{Reid-Miller93}
Margaret Reid-Miller, Gary~L. Miller, and Francesmary Modugno.
\newblock List ranking and parallel tree contraction.
\newblock In John Reif, editor, {\em Synthesis of Parallel Algorithms}, pages
  115--194. Morgan Kaufmann, 1993.

\bibitem{Reif93}
John~H. Reif.
\newblock {\em Synthesis of Paralell Algorithms}.
\newblock Morgan Kaufmann, 1993.

\bibitem{SV81}
Yossi Shiloach and Uzi Vishkin.
\newblock Finding the maximum, merging, and sorting in a parallel computation
  model.
\newblock {\em J. Algorithms}, 2(1), 1981.

\bibitem{shun2015sequential}
Julian Shun, Yan Gu, Guy~E. Blelloch, Jeremy~T Fineman, and Phillip~B Gibbons.
\newblock Sequential random permutation, list contraction and tree contraction
  are highly parallel.
\newblock In {\em {ACM-SIAM} Symposium on Discrete Algorithms (SODA)}, pages
  431--448, 2015.

\bibitem{pam}
Yihan Sun, Daniel Ferizovic, and Guy~E Blelloch.
\newblock Pam: Parallel augmented maps.
\newblock In {\em {ACM} Symposium on Principles and Practice of Parallel
  Programming (PPOPP)}, 2018.

\bibitem{tang2015cache}
Yuan Tang, Ronghui You, Haibin Kan, Jesmin~Jahan Tithi, Pramod Ganapathi, and
  Rezaul~A Chowdhury.
\newblock Cache-oblivious wavefront: improving parallelism of recursive dynamic
  programming algorithms without losing cache-efficiency.
\newblock In {\em {ACM} Symposium on Principles and Practice of Parallel
  Programming (PPOPP)}, 2015.

\bibitem{Tarjan83}
Robert~Endre Tarjan.
\newblock {\em Data Structures and Network Algorithms}.
\newblock Society for Industrial and Applied Mathematics, Philadelphia, PA,
  USA, 1983.

\bibitem{TPL}
https://msdn.microsoft.com/en-us/library/dd460717\%28v=vs.110\%29.aspx.

\bibitem{Valiant91}
Leslie~G. Valiant.
\newblock General purpose parallel architectures.
\newblock In Jan van Leeuwen, editor, {\em Handbook of Theoretical Computer
  Science (Vol. A)}, pages 943--973. MIT Press, 1990.

\bibitem{Vishkin84}
Uzi Vishkin.
\newblock Randomized speed-ups in parallel computation.
\newblock In {\em {ACM} Symposium on Theory of Computing (STOC)}, pages
  230--239, 1984.

\bibitem{Vishkin93}
Uzi Vishkin.
\newblock Advanced parallel prefix-sums, list ranking and connectivity.
\newblock In John Reif, editor, {\em Synthesis of Parallel Algorithms}, pages
  215--257. Morgan Kaufmann, 1993.

\bibitem{Wyllie79}
James~C. Wyllie.
\newblock The complexity of parallel computations.
\newblock Technical Report TR-79-387, Department of Computer Science, Cornell
  University, Ithaca, NY, August 1979.

\end{thebibliography}
